\documentclass[webpdf,contemporary,large]{oup-authoring-template}%

\graphicspath{{Fig/}}

\usepackage{multirow}

\usepackage{wrapfig}

\usepackage{graphicx}
\usepackage[]{geometry}
\usepackage{xcolor}
\definecolor{darkgreen}{RGB}{34, 139, 34} 
\definecolor{ipcolor}{RGB}{242,170,60}
\definecolor{bpcolor}{RGB}{55,126,247}
\usepackage{xspace}
\usepackage{amssymb,amsmath,epsfig}
\usepackage{amsthm}
\usepackage{hyperref}  
\hypersetup{
    colorlinks=true,
    linkcolor=blue, 
    urlcolor=blue,  
    citecolor=blue, 
}

\usepackage{verbatim}
\renewcommand{\cite}[1]{\citep{{#1}}}  

\usepackage{hyperref}\hypersetup{colorlinks=true, citecolor=[rgb]{0,0,1}, urlcolor=black}
\usepackage{pdfcomment}
\usepackage{courier}
\usepackage{float}
\usepackage[skip=4pt]{caption}
\usepackage{subcaption}
\usepackage{tikz}
\usetikzlibrary{shapes,snakes}
\usepackage[capitalise]{cleveref}
\usepackage{tablefootnote}
\usepackage{tabularx}
\usepackage{seqsplit} 

\usepackage{forest}
\usetikzlibrary{arrows.meta,positioning}

\usepackage{booktabs}
\usepackage{siunitx}

\theoremstyle{thmstyleone}%
\newtheorem{theorem}{Theorem}
\theoremstyle{thmstyletwo}%
\theoremstyle{thmstylethree}%
\newtheorem{definition}{Definition}
\newtheorem{corollary}{Corollary}

\newcommand{\nucA}{\ensuremath{\text{\sc A}}}
\newcommand{\nucC}{\ensuremath{\text{\sc C}}}
\newcommand{\nucG}{\ensuremath{\text{\sc G}}}
\newcommand{\nucU}{\ensuremath{\text{\sc U}}}
\newcommand{\pairs}{\ensuremath{\mathit{pairs}}\xspace}
\newcommand{\unpaired}{\ensuremath{\mathit{unpaired}}\xspace}
\newcommand{\MFE}{\ensuremath{\text{\rm MFE}}\xspace}
\newcommand{\UMFE}{\ensuremath{\text{\rm uMFE}}\xspace}

\newcommand{\LP}{\ensuremath{\mathit{loops}}\xspace}
\newcommand{\CR}{\ensuremath{\mathit{critical}}\xspace}

\newcommand{\D}{\ensuremath{\mathrm{\Delta}}}
\newcommand{\DG}{\ensuremath{\mathrm{\Delta} G^\circ}}
\newcommand{\DDG}{\ensuremath{\mathrm{\Delta\Delta} G^\circ}}

\newcommand{\blangle}{\ensuremath{\boldsymbol{\langle}}}
\newcommand{\brangle}{\ensuremath{\boldsymbol{\rangle}}}
\newcommand{\proj}{\ensuremath{\vdash}}

\DeclareMathOperator{\defeq}{\stackrel{\Delta}{=}}

\renewcommand{\vec}[1]{\ensuremath{\boldsymbol{{#1}}}\xspace}

\newcommand{\vecx}{\ensuremath{\vec{x}}\xspace}
\newcommand{\vecy}{\ensuremath{\vec{y}}\xspace}
\newcommand{\vecz}{\ensuremath{\vec{z}}\xspace}

\newcommand{\ystar}{\ensuremath{\vec{y^\star}}\xspace}

\newcommand{\loops}{\ensuremath{\mathit{loops}}\xspace}

\newcommand{\um}{undesignable motif\xspace}
\newcommand{\m}{motif\xspace}
\newcommand{\M}{\ensuremath{\vec{m}}\xspace}
\newcommand{\vecm}{\ensuremath{\vec{m}}\xspace}
\newcommand{\ipairs}{\ensuremath{\mathit{ipairs}}\xspace}
\newcommand{\bpairs}{\ensuremath{\mathit{bpairs}}\xspace}
\newcommand{\mstar}{\ensuremath{\vec{m^\star}}\xspace}

\renewcommand{\emptyset}{\varnothing}
\newcommand{\card}{\ensuremath{\mathit{card}}\xspace}

\newcommand{\pbound}{\ensuremath{\mathit{pbound}}\xspace}
\newcommand{\maxloop}{\ensuremath{\mathit{max\_loop}}\xspace}
\newcommand{\maxwidth}{\ensuremath{\mathit{max\_width}}\xspace}
\newcommand{\maxdepths}{\ensuremath{\mathit{max\_depth}}\xspace}

\newcommand{\highlight}[1]{\emph{#1}}
\newcommand{\blank}{\phantom{*}} 

\newif\ifmarked
\markedtrue         
\markedfalse        

\newcommand{\editcolor}{\ifmarked\color{red}\else\color{black}\fi}

\ifmarked
  \DeclareRobustCommand{\finaledit}[1]{\textcolor{red}{#1}}
  \newcommand{\finaleditcite}[1]{{\hypersetup{citecolor=red}\cite{#1}}}
\else
  \DeclareRobustCommand{\finaledit}[1]{#1}
  \newcommand{\finaleditcite}[1]{\cite{#1}}
\fi
\newsavebox{\looptreebox}
\newsavebox{\approxfigonebox}
\newsavebox{\approxfigtwobox}

\begin{document}

\journaltitle{Journal Title Here}
\DOI{DOI HERE}
\copyrightyear{2022}
\pubyear{2019}
\access{Advance Access Publication Date: Day Month Year}
\appnotes{Paper}

\firstpage{1}


\title[Probabilistic RNA Designability Quantification]{Probabilistic RNA Designability via Interpretable Ensemble Approximation and Dynamic Decomposition}

\author[1]{Tianshuo Zhou\ORCID{0009-0008-4804-0825}}
\author[3,4,5]{David H. Mathews\ORCID{0000-0002-2907-6557}}
\author[1,2,$\ast$]{Liang Huang\ORCID{0000-0001-6444-7045}}

\authormark{Zhou et al.}

\address[1]{\orgdiv{School of EECS}}
\address[2]{\orgdiv{Dept.~of Biochemistry \& Biophysics}, \orgname{Oregon State University}, \country{USA}}
\address[3]{\orgdiv{Dept.~of Biochemistry \& Biophysics}}
\address[4]{\orgdiv{Center for RNA Biology}}

\address[5]{\orgdiv{Dept.~of Biostatistics and Computational Biology}, \orgname{University of Rochester Medical Center}, \country{USA}}

\corresp[$\ast$]{Corresponding author: \href{email:liang.huang.sh@gmail.com}{liang.huang.sh@gmail.com}}

\received{Date}{0}{Year}
\revised{Date}{0}{Year}
\accepted{Date}{0}{Year\vspace{-.5cm}}



\abstract{
\textbf{Motivation:} RNA design, \finaledit{also known as RNA inverse folding}, aims to find RNA sequences that fold into a  target secondary structure.
However, \finaledit{recent work has shown that 
some target structures are provably {\em undesignable},
where no RNA sequence can fold into it as 
the \emph{minimum free energy} (MFE) structure.
In this paper, we go beyond this binary, MFE-based designability
and explore a soft, probability-based 
designability that upperbounds the Boltzmann probability 
of any design and quantifies how easily or likely
any design might possibly fold into the target structure.}
We introduce a theory of ensemble approximation and a probability decomposition framework for bounding the folding probabilities of RNA structures \finaledit{and  motifs} in an explainable way. 
\finaledit{We further develop a linear-time dynamic programming algorithm that efficiently searches over exponentially many decompositions. 
Combining ensemble approximation with dynamic decomposition search, our method efficiently identifies the optimal motif decomposition that yields the tightest probabilistic bound for a given structure.
Our framework is applicable to any {\em factorizable} energy model or scoring function that decomposes onto loops.} \\
\textbf{Results:} Applying our work, \finaledit{\em LinearDecompose}, to both native and artificial RNA structures in the ArchiveII and Eterna100 datasets, we obtained \finaledit{much tighter} probability bounds  than \finaledit{baselines}.
Our work also provides anatomical tools for analyzing RNA structures and \finaledit{pinpointing} the sources of design difficulty at the motif level. \\
 \textbf{Availability:} Source code and data are available at {\url{https://github.com/shanry/RNA-Undesign}}. \\
 \textbf{Supplementary information:} Supplementary text and data are available in a separate PDF.
}
\keywords{RNA Design, RNA Designability, RNA Inverse Folding, Equilibrium Probability}


\maketitle

\section{Introduction}
RNA secondary structures play crucial roles in the functions of non-coding RNAs such as rRNA~\cite{doudna2002chemical} and tRNA~\cite{coller2024trna}. 
RNA design, 
also known as RNA inverse folding, aims to find one or more RNA sequences that fold into a given target structure, typically using the default Turner energy model~\cite{mathews+turner:2006, turner+:2010nndb}.
However, not all RNA structures are designable~\cite{aguirre+:2007computational, yao:2021thesis, zhou+:2024undesignable, zhou+:2025scalable, zhou+:2026:jcb, malik+:2026}. 
Our recent works RIGEND~\cite{zhou+:2024undesignable} and FastMotif~\cite{zhou+:2025scalable} proved that many native structures in the ArchiveII dataset and many artificial structures in the Eterna100 benchmark are undesignable, \finaledit{in the sense
that no sequence can ever fold into them as MFE structures.} However, these methods only study the
\MFE-based criteria and therefore only provide  a {\em binary} notion of designability. As a result, they cannot address the \finaledit{more important ensemble-level design objectives such as equilibrium probability
which, by considering competing structures, is more appropriate for RNA design}~\finaleditcite{Zadeh+:2010, ward+:2023fitness, zhou+:2023samfeo, tang+:2026}.

This perspective motivates a shift from binary \MFE-based criteria to a probabilistic characterization of designability. 
While \MFE-based undesignability can certify that a structure is impossible to be a minimum free energy structure of any RNA sequence, it cannot quantify \emph{how likely} or \finaledit{how easily}
a sequence might possibily fold into that structure.
For example, consider two structures $\vecy_1$ and $\vecy_2$ that are both undesignable under the \MFE
  criterion. It is nevertheless possible that there exists a sequence $\vecx_1$ where~$p(\vecy_1 \,|\, \vecx_1)=0.2$, whereas no sequence can make~$p(\vecy_2 \,|\, \vecx)$ exceed $0.01$. This discrepancy is invisible to \MFE-based analysis but is crucial for understanding practical designability of a target structure. 
  These observations motivate the study of \emph{probabilistic designability}, defined as an upper bound on $p(\ystar \mid \vecx)$ \finaledit{over all sequences} \vecx
  for a given target structure $\ystar$. The tighter the upper bound, the better we understand the 
  \finaledit{designability of that structure}.

Despite its importance, probabilistic designability has received little attention,
with CountingDesign~\cite{yao+:2019,yao:2021thesis} \finaledit{being the only} existing method that addresses this question. However, it relies on exhaustive enumeration: for a given RNA structure or  motif, it enumerates \finaledit{and folds} all (partial) sequences that satisfy its base-pairing constraints and record the maximum equilibrium probability observed. This brute-force strategy is neither scalable nor interpretable, requiring weeks of computation for structures or motifs of length up to 14.

To bridge the gap between \MFE-based and probabilistic \finaledit{designabilities}, we propose two \finaledit{main ideas}: \emph{ensemble approximation} and \emph{probability decomposition}.
Ensemble approximation represents the full RNA folding ensemble using a small, explicitly identified set of rival structures (motifs) that compete with the target structure. This approximated ensemble enables us to derive rigorous upper bounds on folding probability while maintaining full explainability, as each bound is supported by concrete rival structures that thermodynamically dominate the target.
Furthermore, we prove that the probability bound of a target structure is no greater than the \emph{product of the local probability bounds} of its constituent motifs under {\em any} structural decomposition. 
\finaledit{In order to explore as many decompositions 
as possible},
we develop a {\em linear-time} dynamic programming algorithm that  \finaledit{explores} {\em exponentially many} decompositions to find the one \finaledit{with} the tightest upper bound.
Our contributions are:
\begin{enumerate}
\item \textbf{Theory.} We propose a novel and elegant theory of ensemble approximation to interpretably characterize probabilistic RNA designability. We also prove that the probability bound of a structure can be \finaledit{factored into products of} local bounds of structural motifs, revealing a more nuanced connection between local and global \finaledit{designabilities}.
\item \textbf{Algorithms.} We develop efficient algorithms \finaledit{(Algs.~\ref{alg:approxe_1} and \ref{alg:approxe_2})} for approximating Boltzmann ensembles using rival structures or motifs, and introduce a linear-time dynamic programming approach \finaledit{(Algorithm~\ref{alg:topdown-dp})}  that efficiently explores exponentially many decompositions to obtain the \finaledit{tightest probability bound}.
\item \textbf{Application.} 
Applying our methods to both native and artificial RNA structures in the ArchiveII and Eterna100 \finaledit{datasets}, we obtained probability bounds that are much tighter than prior approaches.
In addition, our methods further provide anatomical tools for analyzing RNA structures and understanding the sources of design difficulty at the motif level.
\end{enumerate}

\vspace{-0.8cm}
\section{\finaledit{Preliminaries: RNA Structures and Motifs}}\label{sec2}

\begin{figure}[t]
    \centering
    \pdftooltip{%
    \includegraphics[width=0.7\linewidth]{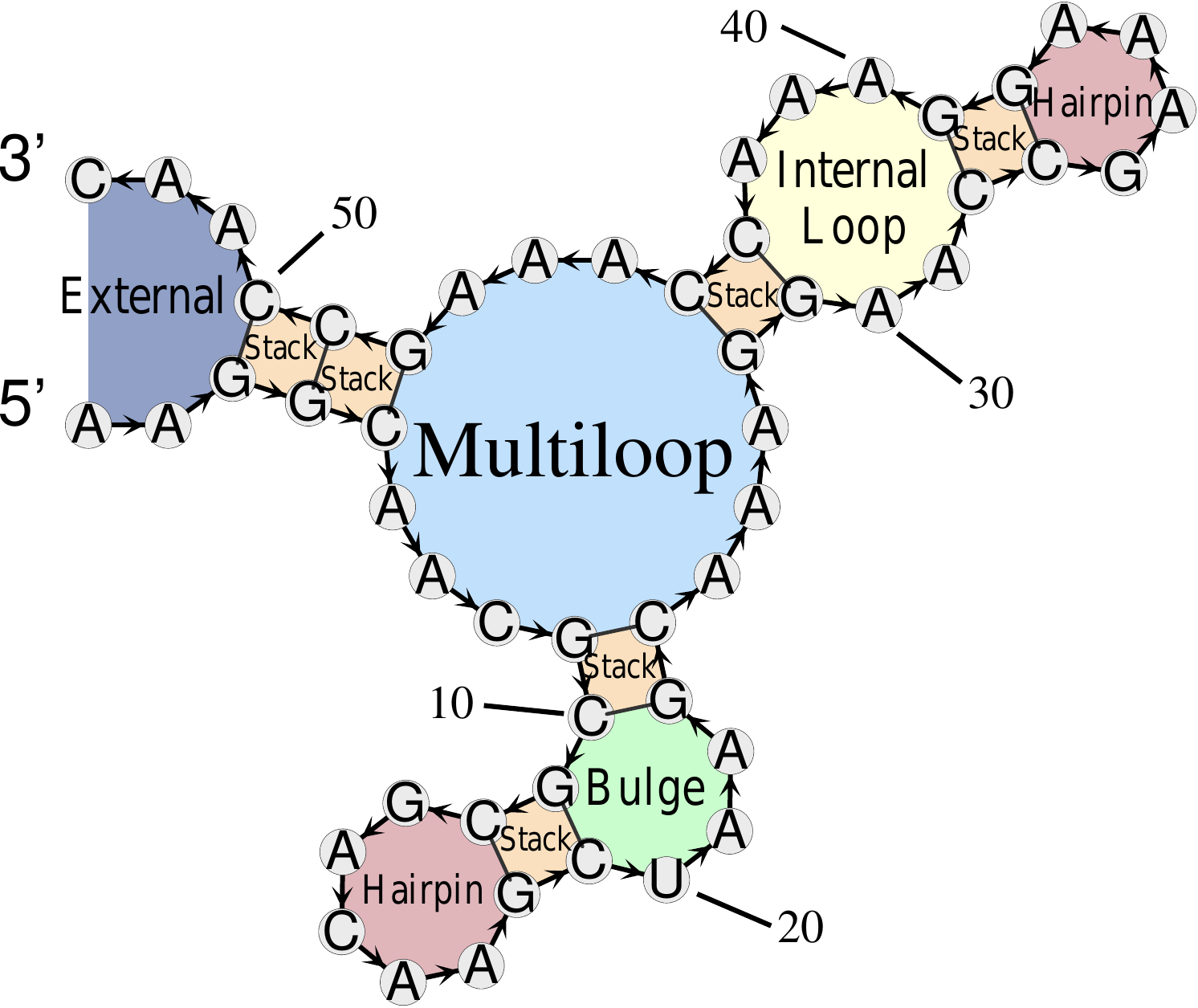}%
    }{Diagram of an RNA secondary structure annotated with its loop types, including hairpin loops, internal loops, bulges, stacks, a multiloop, and an external loop, with arrows indicating 5' to 3' direction.}
    \caption{A secondary structure and its loops \finaledit{(including stacks)}.}
    \vspace{-0.3cm} 
    \label{fig:loop}
\end{figure}
An RNA sequence $\vecx$ of length $n$ is a string of  nucleotides $\vecx_1\vecx_2\dots \vecx_n$,  where $\vecx_i \in \{\nucA, \nucC, \nucG, \nucU\}$. 
A secondary structure~$\mathcal{P}$  for $\vecx$ is a set of paired indices where each pair $(i, j) \in \mathcal{P}$ indicates two distinct bases $\vecx_i \vecx_j \in \{\nucC\nucG,\nucG\nucC,\nucA\nucU,\nucU\nucA, \nucG\nucU,\nucU\nucG\}$ and \finaledit{each base 
 only be paired at most once}. A secondary structure is pseudoknot-free \finaledit{(i.e., nested pairs only)} if there are no two pairs $(i, j)\in \mathcal{P}\text{ and }(k, l)\in~\mathcal{P}$ such that  $i \!<\! k \!<\! j \!<\! l$. \finaledit{This work does not consider pseudoknots.}
 $\mathcal{P}$ can also be represented as a \finaledit{dot-bracket} string~$\vecy=\vecy_1\vecy_2\dots \vecy_n$,  where a pair of indices $(i, j) \in\mathcal{P}$ corresponds to $\vecy_i=``("$ and $\vecy_j=``)"$ and any unpaired index $k$ corresponds to $\vecy_k=``."$. The unpaired indices in $\vecy$ are denoted as $\unpaired(\vecy)$ and the set of paired indices in $\vecy$ is denoted as $\pairs(\vecy)$, which  equals~$\mathcal{P}$. 
\finaledit{For more details, see LinearFold~\cite{huang+:2019} and LinearPartition~\cite{zhang+:2020}.}

The \emph{ensemble} of an RNA sequence $\vecx$  is the set of all secondary structures that  $\vecx$ can possibly fold into, denoted as $\mathcal{Y}(\vecx)$. The \highlight{free energy (change)} $\DG(\vecx, \vecy)$   to characterizes the stability of $\vecy \in \mathcal{Y}(\vecx)$. The lower 
$\DG(\vecx, \vecy)$, 
the more stable the  structure $\vecy$ for $\vecx$. 

A structure \vecy is composed of a set of loops denoted as $\LP(\vecy)$ \finaledit{where stacks of adjacent base pairs are also considered special loops}, as illustrated in Fig.~\ref{fig:loop}. 
 The free energy change of a secondary structure $\vecy$ is the sum of the free energy change of each loop, i.e., 
 \vspace{-0.2cm}
 \begin{equation}
 \DG(\vecx, \vecy) = \sum_{\vecz \in \loops(\vecy)} \DG(\vecx, \vecz), \label{eq:free_energy}
 \end{equation}
where each term $ \DG(\vecx, \vecz)$ is the energy for loop \vecz. The energy of each loop is typically determined by nucleotides on the positions of enclosing pairs and their
adjacent mismatch positions, which are named as \emph{critical positions} and denoted as $\CR(\vecz)$.
See {Supplementary Section~\ref{sec:loops}} for a detailed explanation of different loop types and their critical positions.
The structure with the \highlight{minimum free energy} is the most stable structure in the ensemble $\mathcal{Y}(\vecx)$. The minimum free energy of $\mathcal{Y}(\vecx)$ is 
defined as
\vspace{-0.2cm}
\begin{equation}
\MFE(\vecx)=\min_{\vecy \in \mathcal{Y}(\vecx)} \DG(\vecx, \vecy). \label{eq:mfe_df}
\end{equation}
\finaledit{The {partition function} sums up the Boltzmann terms of all structures,}
\begin{equation}
\finaledit{Q(\vecx)=\sum_{\vecy \in \mathcal{Y}(\vecx)}e^{-\DG(\vecx, \vecy)/RT},}
\end{equation}
\finaledit{which defines
the equilibrium probability or Boltzmann} probability of a structure $\vecy$ in the ensemble
\begin{equation}
p(\vecy \mid \vecx) = \frac{e^{-\DG(\vecx, \vecy)/RT}}{Q(\vecx)} = \frac{e^{-\DG(\vecx, \vecy)/RT}}{\sum_{\vecy' \in \mathcal{Y}(\vecx)}e^{-\DG(\vecx, \vecy')/RT}}. \label{eq:prob}
\end{equation}

 A (structural) \highlight{motif} is defined as a contiguous set of loops in a structure, which 
 \finaledit{generalizes RNA structures} (See Supplementary Section~\ref{sec:motif} or FastMotif~\cite{zhou+:2024undesignable} for details).
 The thermodynamic definitions in \cref{eq:free_energy,eq:mfe_df,eq:prob} can be naturally extended to motifs. For example, the equilibrium probability of a motif \vecm in the ensemble of \finaledit{a (partial) sequence} \vecx is
\begin{equation}
p(\vecm \mid \vecx) = \frac{e^{-\DG(\vecx, \vecm)/RT}}{Q(\vecx)} = \frac{e^{-\DG(\vecx, \vecm)/RT}}{\displaystyle\sum_{\vecm' \in \mathcal{M}(\vecx)}e^{-\DG(\vecx, \vecm')/RT}}, \label{eq:mprob}
\end{equation}
where $\mathcal{M}(\vecx)$ is the motif ensemble of  (partial) sequence \vecx.

\vspace{-0.3cm}
\section{\finaledit{RNA Designability: MFE vs.~Probabilistic}}

\finaledit{We first review the MFE-based designability before moving on to probabilistic designability.}

\vspace{-0.3cm}
\subsection{\MFE-based Designabiity}
Given a target structure $\ystar$, \MFE-based RNA design aims to find a suitable RNA sequence $\vecx$ such that $\ystar$ is an \MFE structure of $\vecx$. 
We define $\mathcal{X}(\vecy)$ to be the \finaledit{design space} of structure \vecy,
which is the set of all RNA sequences 
that might fold into \vecy,
i.e., 
$\mathcal{X}(\vecy)=\{ \vecx \mid \vecy \in \mathcal{Y}(\vecx) \}$.
Structure $\ystar$ is undesignable by the \MFE criterion if and only if
\begin{equation}
 ~\forall \vecx \in \mathcal{X}(\vecy),  \exists \vecy' \ne \ystar , \DG(\vecx, \vecy') <  \DG(\vecx, \vecy), \label{eq:mfe}
\end{equation}
\finaledit{which means for any sequence \vecx there is always a rival structure $\vecy'$ that is energetically more favored than \ystar.}

\finaledit{Similarly,} following the unique \MFE (\UMFE) criterion from previous studies~\cite{halevs+:2015combinatorial,yao+:2019,ward+:2023fitness,zhou+:2023samfeo}, a structure $\ystar$ is undesignable by the (\UMFE) criterion if and only if
\begin{equation}
 ~\forall \vecx \in \mathcal{X}(\vecy),  \exists \vecy' \ne \ystar , \DG(\vecx, \vecy') \leq  \DG(\vecx, \vecy), \label{eq:umfe}
\end{equation}
\finaledit{which means for any sequence \vecx there is always a rival structure $\vecy'$ that is energetically no worse than \ystar.}

\vspace{-0.2cm}
\subsection{Probabilistic Designability}
\MFE-based designability cannot quantify {how likely} 
or \finaledit{how easily} a target structure can possibly form in the ensemble of any RNA sequence. To address this problem, we aim to find an upper bound $\pbound(\vecy)$ such that 
\vspace{-0.2cm} 
\[\finaledit{\pbound(\vecy) \geq \max_{\vecx} p(\vecy \mid \vecx).}\]
While it is hard to obtain the exact value of $\max_{\vecx} p(\vecy \,|\, \vecx)$, the tighter $\pbound(\vecy)$ is, the better we know about the probabilistic designability of \vecy.
Similarly, we can also find an upper bound $\pbound(\vecm)$ for a motif \vecm such that $\pbound \geq (\vecm)\max_{\vecx} p(\vecm \mid \vecx) $.

\vspace{-0.5cm}

\section{Ensemble Approximation via Rival Search}\label{sec3}

Rival structures (motifs) play a pivotal role in the previous works RIGEND~\cite{zhou+:2024undesignable} and FastMotif~\cite{zhou+:2025scalable}, which can prove RNA structures (motifs) undesignable under the \UMFE criterion.
The central idea is to identify a small set of rival structures (motifs)
$$Y_r = \{\vecy_1, \vecy_2, \ldots, \vecy_k\}, \ystar \notin Y_r,$$
such that
\begin{equation}
\forall \vecx, \exists \vecy'  \in Y_r,   \DG(\vecx, \vecy') \leq \DG(\vecx, \ystar). \label{eq:rival} 
\end{equation} 
In other words, for every RNA sequence \vecx, at least one explicitly identified rival structure is energetically no worse than the target structure. This condition guarantees that the target structure can never be the unique minimum free energy conformation of any sequence and is therefore \UMFE-undesignable.

While effective, both RIGEND and FastMotif treat rival structures as \emph{independent} competitors. Consequently, the collective thermodynamic effect of multiple rival structures is left unexploited. In this section, we introduce a new theory of Boltzmann ensemble approximation that leverages the collective contribution of rival structures to derive an upper bound on the equilibrium probability of the target structure.

To make the derivation intuitive, we begin with the simplest case, in which the ensemble is approximated by the target structure \finaledit{plus} a single rival structure.
We then extend this theory to the general case, where the ensemble approximation consists of multiple rival structures (motifs).  
Since motifs are generalizations of structures, we present the theory using structures for notational simplicity; all results apply to motifs directly.

\subsection{Single Rival Structure (Motif)}
We illustrate the idea using the Eterna100 structure ``Simple Single Bond'' (we cut off trailing unpaired bases to fit the page). When attempting to design this target structure $\ystar$ (Fig.~\ref{fig:ex1_approxe}), the resulting sequence $\vecx$ consistently folds into a different but structurally very similar conformation $\vecy'$.

\begin{figure}[t]
\centering
\begin{lrbox}{\approxfigonebox}%
\begin{minipage}{\linewidth}%
\footnotesize
\setlength{\tabcolsep}{3pt}
\begin{tabularx}{\linewidth}{@{} r @{:\ } >{\ttfamily\raggedright\arraybackslash}X @{}}
$\ystar$   & \seqsplit{.....{\underline{\color{blue}.(.}}.......{\underline{\color{blue}.(}}(((.....))){\underline{\color{blue}).}}.......{\underline{\color{blue}.).}}..}\\
$\vecy'\,$ & \seqsplit{.....{\underline{\color{blue}...}}.......{\underline{\color{blue}.(}}(((.....))){\underline{\color{blue}).}}.......{\underline{\color{blue}...}}..}\\
$\vecx$    & \seqsplit{AUAAG{\underline{\color{blue}CGG}}UAAAAAA{\underline{\color{blue}AG}}UGCGAAAAGCA{\underline{\color{blue}UG}}AAAAAAA{\underline{\color{blue}ACA}}GA}\\
$\Delta$   & \blank\blank\blank\blank\blank***%
             \blank\blank\blank\blank\blank\blank\blank**%
             \blank\blank\blank\blank\blank\blank\blank\blank\blank\blank\blank**%
             \blank\blank\blank\blank\blank\blank\blank***\\
\end{tabularx}
\end{minipage}%
\end{lrbox}%
\pdftooltip{\usebox{\approxfigonebox}}{Table showing target structure ystar in dot-bracket notation, rival structure y-prime, RNA sequence x (AUAAGCGGUAAAAAAAGUGCGAAAAGCAUGAAAAAAAACAGA), and differential positions Delta marked with asterisks, illustrating ensemble approximation with a single rival structure for the Eterna100 structure Simple Single Bond.}
\caption{Ensemble approximation with a single rival structure.}
\label{fig:ex1_approxe}
\end{figure}
As proven in RIGEND~\cite{zhou+:2024undesignable}, it turns out that the free energy change difference between $\vecy'$ and $\ystar$ 
$$\DDG(\vecx, \vecy', \ystar) \defeq \DG(\vecx, \vecy') - \DG(\vecx, \ystar)$$ 
is only dependent on \finaledit{a number of} \emph{differential positions} \finaledit{$\D(\vecy', \ystar)$ \finaledit{(e.g., pairs and mismatches)}, defined as}
$$\D(\vecy', \ystar) \defeq \bigcup_{\vecz \in \LP(\ystar)\ominus \LP(\vecy')} \CR(\vecz).\footnote{\finaledit{$\ominus$ denotes the symmetric difference between two sets.}}$$
Fig.~\ref{fig:ex1_approxe} annotates the differential positions. 
Consequently, to verify that $\DDG(\vecx, \vecy', \ystar) \!<\!0$ holds
for all sequences \vecx, it suffices to enumerate nucleotide assignments restricted to $\D(\vecy', \ystar)$.\footnote{\finaledit{The maximum size of $\D(\vecy', \ystar)$ or enumeration is constrained to bound complexity or running time.}} This reduction enables an efficient and interpretable proof that $\ystar$ is undesignable under the \UMFE criterion. Beyond the \UMFE undesignability, however, we notice that this process yields additional thermodynamic insight. In particular, rather than merely establishing the sign of $\DDG(\finaledit{\vecx}, \vecy', \ystar)$, we can compute the quantity,
$$\max_{\vecx}  \DDG(\vecx, \vecy', \ystar),$$
which represents the guaranteed energy advantage of the rival structure over the target across all sequences. This quantity directly leads to an upper bound on the equilibrium probability of the target structure: 
\begin{equation}
\begin{aligned}
p(\boldsymbol{y^\star} \mid \boldsymbol{x}) &= \frac{e^{-\DG(\boldsymbol{x}, \boldsymbol{y^\star}) / R T}}{\sum_{\boldsymbol{y} \in \mathcal{Y}(\boldsymbol{x})} e^{-\DG(\boldsymbol{x}, \boldsymbol{y}) / R T}} \\
& \leq \frac{e^{-\DG(\boldsymbol{x}, \boldsymbol{y^\star}) / R T}}{ e^{-\DG(\boldsymbol{x}, \boldsymbol{y^\star}) / R T} + e^{-\DG(\boldsymbol{x}, \boldsymbol{y'}) / R T}} \\
& = \frac{1}{1 + e^{-(\DG(\boldsymbol{x}, \boldsymbol{y'}) - \DG(\boldsymbol{x}, \boldsymbol{y^\star})) / R T}}.
\end{aligned}
\end{equation}
Taking the maximum over all sequences $\vecx$, we obtain the following upper bound:
\begin{equation}
\max_{\vecx} p(\boldsymbol{y^\star} \mid \boldsymbol{x}) \leq \frac{1}{1 + e^{-\max_{\vecx} \DDG(\vecx, \vecy', \ystar) / RT}} .
\end{equation}
 Algorithm~\ref{alg:approxe_1} shows 
how to identify an upper bound for $\max_{\vecx} p(\boldsymbol{y^\star} \mid \boldsymbol{x})$  using a single rival structure \finaledit{$\vecy'$, which can be obtained by first taking a sequence $\vecx$ designed for $\ystar$ and then folding $\vecx$ (assume $\MFE(\vecx)$ is close to \ystar)}.
\finaledit{Note that $\vdash$ denotes the projection of a sequence $\vecx$ onto a set of positions ($\Delta(\vecy', \ystar)$). (See Supp. Sec.~\ref{sec:proj})} 

\begin{algorithm}[t]
\caption{Ensemble Approximation with a Single Rival Structure (Motif)}\label{alg:approxe_1}
\begin{algorithmic}
\vspace{-0.2cm}
\State
\Function{EnsembleApproximation}{$\ystar$, $\vecy'$}
    \State \finaledit{$\displaystyle\DDG_{\max} \gets \!\!\max_{\hat{\vecx} \in \{\vecx \vdash \Delta(\vecy', \ystar) \mid \vecx \in \mathcal{X}(\ystar) \}} \!\!\!\DDG(\hat{\vecx}, \vecy', \ystar)$}
    \State \Return $\frac{1}{1 + e^{-\DDG_{\max} / RT}}$  \Comment{upper bound $pbound(\ystar)$}
\EndFunction
\end{algorithmic}
\end{algorithm}

\subsection{Multiple Rival Structures (Motifs)}
 Algorithm~\ref{alg:approxe_1} is effective and interpretable, \finaledit{but} the resulting bound is limited by the coarse approximation of the ensemble using only one rival structure. 
 Incorporating more rival structures \finaledit{will} yield a tighter approximation, though enumerating the full ensemble is infeasible.

\begin{figure}[b]
\centering
\begin{lrbox}{\approxfigtwobox}%
\begin{minipage}{0.85\linewidth}%
\footnotesize
\setlength{\tabcolsep}{3pt}

\begin{tabularx}{\linewidth}{@{} r @{:\ } >{\ttfamily\raggedright\arraybackslash}X @{}}
$\ystar$
  & \texttt{\seqsplit{....((((({\underline{\color{blue}(((.(....)).).)}}.)))))....}} \\

$\vecy'^1$
  & \texttt{\seqsplit{....((((({\underline{\color{blue}((..(....)..).)}}.)))))....}} \\

$\vecy'^2$
  & \texttt{\seqsplit{....((((({\underline{\color{blue}(((.(....)))..)}}.)))))....}} \\

$\vecy'^3$
  & \texttt{\seqsplit{....((((({\underline{\color{blue}(((((....)))).)}}.)))))....}} \\

$\vecy'^4$
  & \texttt{\seqsplit{....((((({\underline{\color{blue}(((.(....).)).)}}.)))))....}} \\

$\vecy'^5$
  & \texttt{\seqsplit{....((((({\underline{\color{blue}(.(((....)).).)}}.)))))....}} \\

$\vecy'^6$
  & \texttt{\seqsplit{....((((({\underline{\color{blue}(((((...))).).)}}.)))))....}} \\

$\vecy'^7$
  & \texttt{\seqsplit{....((((({\underline{\color{blue}(.(((....)))..)}}.)))))....}} \\

$\vecy'^8$
  & \texttt{\seqsplit{....((((({\underline{\color{blue}(..((....))...)}}.)))))....}} \\

$\vecy'^9$
  & \texttt{\seqsplit{....((((({\underline{\color{blue}(.............)}}.)))))....}} \\

$\Delta$ &~\blank\blank\blank\blank\blank\blank\blank\blank***************\\

\end{tabularx}
\end{minipage}%
\end{lrbox}%
\pdftooltip{\usebox{\approxfigtwobox}}{Table showing target structure ystar and nine rival structures y-prime-1 through y-prime-9 in dot-bracket notation for the Eterna100 structure Zigzag Semicircle, with differential positions Delta spanning 15 consecutive positions marked with asterisks, illustrating ensemble approximation with multiple rival structures.}
\caption{Ensemble approximation with nine rival structures.}
\label{fig:ex2_approxe}
\end{figure}

We therefore generalize the framework to include multiple rival structures. Figure~\ref{fig:ex2_approxe} illustrates this idea using the Eterna100 structure ``Zigzag Semicircle". Our previous work~\cite{zhou+:2024undesignable} proved this structure \UMFE-undesignable by identifying nine rival structures
\[
\forall \vecx, \exists \vecy' \in Y_r = \{\vecy'^1, \vecy'^2,\ldots, \vecy'^9\},  \text{s.t.~}
\DDG(\vecx, \vecy', \ystar)\!\leq\! 0.\notag
\]
Rather than selecting a single rival, we include all nine rivals in the ensemble approximation to obtain a tighter bound. Specifically,

\begin{equation}
\hspace{-0.4cm}
\begin{aligned}
p(\boldsymbol{y^\star} \mid \boldsymbol{x}) &\leq \frac{e^{-\finaledit{\DDG}(\boldsymbol{x}, \boldsymbol{y^\star}) / R T}}{ e^{-\finaledit{\DDG}(\boldsymbol{x}, \boldsymbol{y^\star}) / R T} + \sum_{\vecy' \in Y_r} e^{-\finaledit{\DDG}(\boldsymbol{x}, \boldsymbol{y'}) / R T} } \hspace{-1cm}\\
& \leq \frac{1}{1 + \sum_{\vecy' \in Y_r} e^{-(\finaledit{\DDG}(\boldsymbol{x}, \boldsymbol{y'}) - \finaledit{\DDG}(\boldsymbol{x}, \boldsymbol{y^\star})) / R T} } \hspace{-1cm}\\
& \leq \frac{1}{1 + \sum_{\vecy' \in Y_r} e^{-\DDG(\boldsymbol{x}, \boldsymbol{y'}, \ystar) / R T} }  \hspace{-1cm}\\
& \leq \frac{1}{1 +  e^{-\DDG(\boldsymbol{x}, Y_r, \ystar) / R T} },
\end{aligned}
\hspace{-0.3cm}
\end{equation}
where $e^{-\DDG(\boldsymbol{x}, Y_r, \ystar) / R T} \finaledit{\defeq \sum_{\vecy' \in Y_r} e^{-\DDG(\boldsymbol{x}, \boldsymbol{y'}, \ystar) / R T}}$ summarizes the collective energetic dominance of the rival set over the target structure. Taking the maximum over all sequences yields
\begin{equation}
\max_{\vecx} p(\boldsymbol{y^\star} \mid \boldsymbol{x}) \leq \frac{1}{1 + e^{-\max_{\vecx} \DDG(\vecx, Y_r, \ystar) / RT}} .
\end{equation}

As in the single-rival case, exhaustive enumeration over all sequences is unnecessary. Owing to the sparsity of loop energy contributions in the Turner model, it suffices to enumerate nucleotides on the \emph{overall differential positions}
\begin{equation}
\D(Y_r, \ystar) = \cup_{\vecy' \in Y_r} \D(\vecy', \ystar).\label{eq:dp_multiple}
\end{equation}

Algorithm~\ref{alg:approxe_2} outlines the resulting procedure to compute the upper bound using multiple rival structures.

\begin{algorithm}[b]
\caption{Ensemble Approximation with Multiple Rival Structures (Motifs)}\label{alg:approxe_2}
\begin{algorithmic}
\vspace{-0.2cm}
\State
\Function{EnsembleApproximation}{$\ystar$, $Y_r$}
    \State \finaledit{$\displaystyle\DDG_{\max} \gets 
    \!\! \max_{\hat{\vecx} \in \{\vecx \vdash \Delta(Y_r, \ystar) \mid \vecx \in \mathcal{X}(\ystar) \}} 
    \!\!\DDG(\hat{\vecx}, Y_r, \ystar)$}
    \State \Return $\frac{1}{1 + e^{-\DDG_{\max} / RT}}$  \Comment{upper bound $pbound(\ystar)$}
\EndFunction
\end{algorithmic}
\end{algorithm}


\subsection{Complexity Analysis}
Algorithm~\ref{alg:approxe_1} is a special case of Algorithm~\ref{alg:approxe_2} with a single rival structure.
The computational complexity of Algorithm~\ref{alg:approxe_2} is determined by the size of the differential positions $\D(Y_r,\ystar)$ \finaledit{defined in Eq.~\ref{eq:dp_multiple}}, as the algorithm enumerates all nucleotide assignments restricted to these positions.
Specifically, the time complexity is
\begin{equation}
\editcolor
O\!\left(6^{|\pairs(\D(Y_r, \ystar))|} \cdot 4^{|\unpaired(\D(Y_r, \ystar))|}\right),
\end{equation}
where paired positions admit six canonical base-pair types and unpaired positions admit four nucleotides.

In practice, to control runtime, we sample rival structures (motifs) by folding sequences compatible with the target structure (motif) and impose an upper limit on the number of differential assignments enumerated.
If this limit is exceeded and no suitable rival structures can be sampled, the corresponding structure or motif is skipped, trading completeness for efficiency.

\section{Linear-time Dynamic Programming over Exponentially Many Decompositions}\label{sec4}

\subsection{Structure and Probability Decomposition}
\begin{figure}[t]
  \centering
  \pdftooltip{\includegraphics[width=0.55\linewidth]{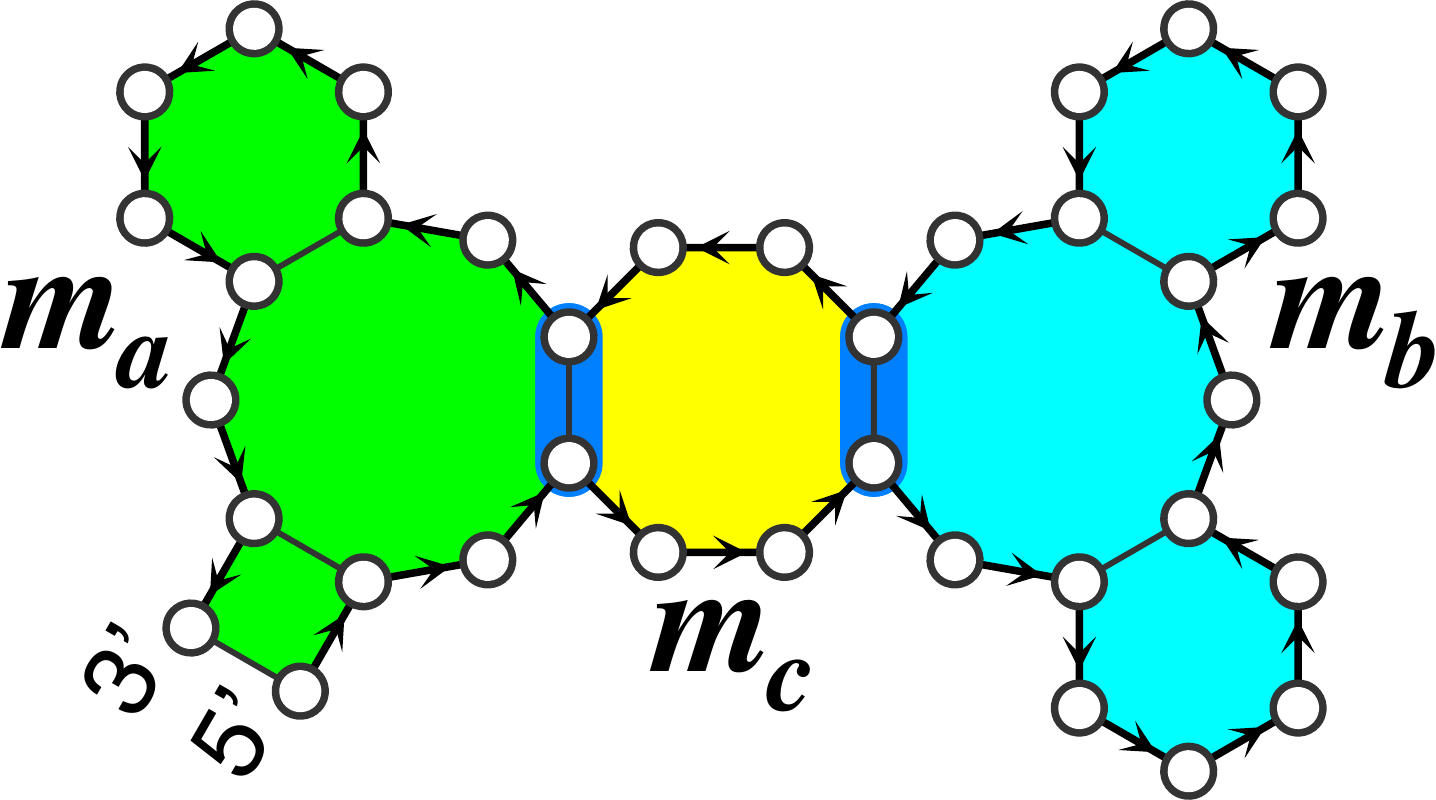}}{Diagram of the Eterna100 RNA structure ``multilooping fun'' decomposed into three non-overlapping motifs $m_a$, $m_b$, and $m_c$, each highlighted in a different color.}
\caption{Example of structure decomposition for the Eterna100 structure ``multilooping fun", which is decomposed into 3 \finaledit{motifs} $\vecm_a, \vecm_b, \text{and } \vecm_c$ highlighted in different colors.}
\label{fig:decomposition_eterna57}
\end{figure}
\begin{figure}
  \centering
  \pdftooltip{\includegraphics[width=0.4\linewidth]{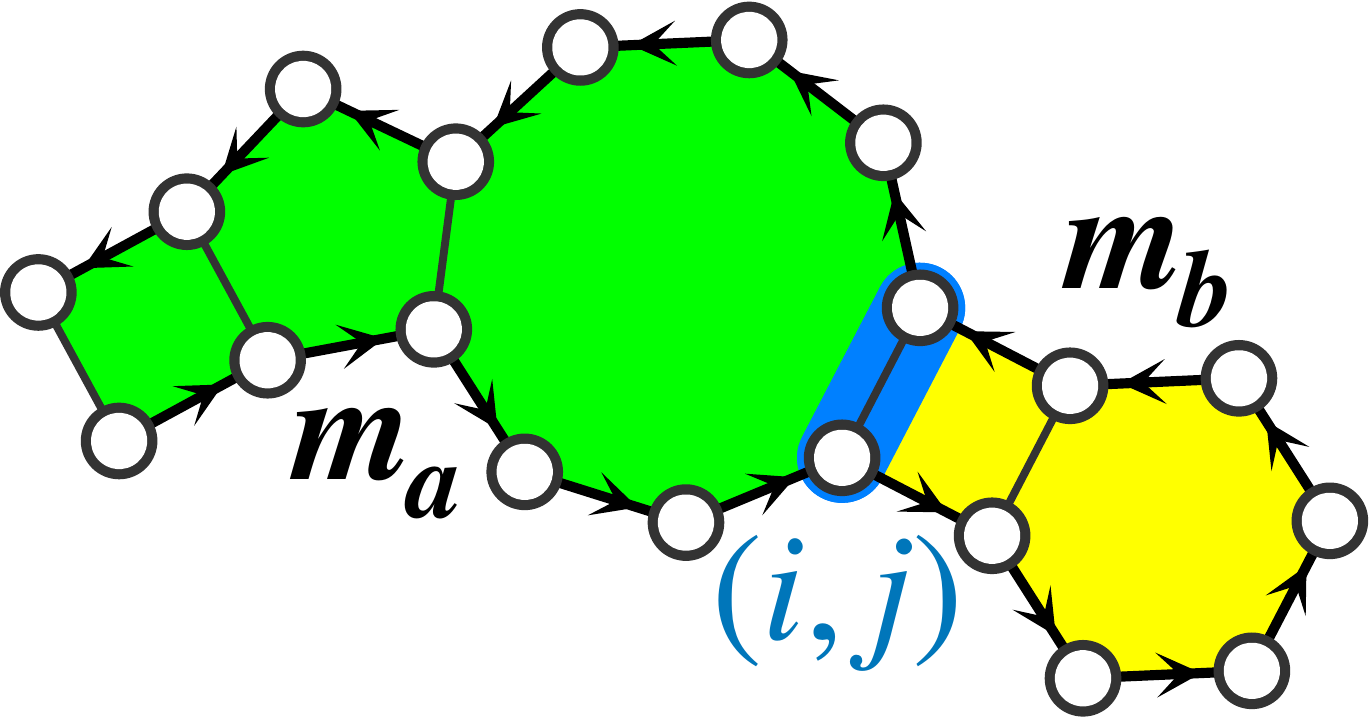}}{Diagram illustrating a motif m being split into two submotifs $m_a$ and $\m_b$ at a boundary base pair $(i, j)$.}
\caption{A motif \vecm is split into $\vecm_a$ and $\vecm_b$ at the base pair $(i, j)$}
\label{fig:split}
\end{figure}

FastMotif~\cite{zhou+:2025scalable} shows that \MFE-based structure undesignability can often be attributed to \emph{local} (minimal) undesignable motifs.
For example, Fig.~\ref{fig:decomposition_eterna57} decomposes the Eterna100 structure \emph{multilooping fun} into three motifs $\vecm_a$, $\vecm_b$, and $\vecm_c$; in this instance, either $\vecm_a$ or $\vecm_b$ alone can certify the \UMFE-undesignability of the full structure, providing a more localized and interpretable explanation than relying solely on global rival structures.


The principle that global designability is constrained by local designability also extends to \emph{probabilistic} designability.
CountingDesign~\cite{yao:2021thesis} proved that for any motif $\vecm$ contained in a target structure $\ystar$, the motif probability upper-bounds the structure probability:
 \begin{equation}
 \forall \vecm \in \ystar, p(\vecm \mid \vecx) \finaledit{\;\geq\; }  p(\ystar \mid \vecx).
 \end{equation}
Applying this result to Fig.~\ref{fig:decomposition_eterna57} yields, for example,
$p(\ystar \mid \vecx) \le p(\vecm_a \mid \vecx)$ and $p(\ystar \mid \vecx) \le p(\vecm_b \mid \vecx)$.

However, CountingDesign \finaledit{bounds} $p(\ystar \,|\, \vecx)$ by the probability of a \emph{single} motif. The joint impact of multiple non-overlapping motifs remains unexplored.
\finaledit{Here} we show that when a structure can be decomposed into non-overlapping motifs, the structure probability is upper-bounded by the \emph{product} of the motif probabilities.

\begin{theorem}[Probability decomposition over non-overlapping motifs]
If a \finaledit{pseudoknot-free} structure $\ystar$ can be decomposed into a set of non-overlapping motifs
$\boldsymbol{M}=\{\vecm_1,\vecm_2,\ldots,\vecm_C\}$,
then for any sequence $\vecx \in \mathcal{X}(\ystar)$,
\begin{equation}
p(\ystar \mid \vecx) \le \prod_{\vecm \in \boldsymbol{M}} p(\vecm \mid \vecx).
\label{eq:prob_decomposition}
\end{equation}
\end{theorem}

\begin{proof}
As a \finaledit{pseudoknot-free} structure can be decomposed recursively, it suffices to prove the binary splitting step.
Suppose a motif $\vecm$ is split into two submotifs $\vecm_a$ and $\vecm_b$ at a boundary base pair $(i,j)$, as illustrated in Fig.~\ref{fig:split}.
Let $P_{i,j}$ denote the event that bases $i$ and $j$ are paired.
By the chain rule of probability,
\begin{equation}
\begin{aligned}
&\;\quad  {p} ( \vecm \mid \vecx ) = {p} ( \vecm_a +  \vecm_b \mid \vecx ) \\
&=   {p} (P_{i,j} \mid \vecx) \times {p} ( \vecm_a \mid \vecx, P_{i,j} )  \times {p} ( \vecm_b \mid \vecx, P_{i,j}) \\
&=   {p} (P_{i,j} \mid \vecx) \times {p} ( \vecm_a \mid \vecx)  \times {p} ( \vecm_b \mid \vecx) \\
&\leq  p( \vecm_a \mid \vecx)  \times {p} ( \vecm_b \mid \vecx). \label{eq:product}
\end{aligned}
\end{equation}
\finaledit{In a pseudoknot-free structure, the boundary pair $(i, j)$ ensures that the two submotifs $\vecm_a$ and $\vecm_b$ are independent (no interactions between them) in the Boltzmann ensemble, which justifies the third line.
The fourth line holds because the boundary pairs are enforced in the motif-level folding.}
\end{proof}
\begin{corollary}
\begin{equation}
 \pbound(\ystar) \leq \prod_{\boldsymbol{m} \in \boldsymbol{M}} \pbound(\boldsymbol{m}).\label{eq:prob_decomposition}
\end{equation} \end{corollary}


This \finaledit{corollary} suggests that we can tighten the probability upper bound for a structure by exploiting probability bounds of smaller local motifs.
For sufficiently small motifs, we can obtain an exact bound by enumerating all compatible (partial) sequences and folding them under the motif constraints.
For motifs that are too large for exhaustive enumeration, we instead apply the motif ensemble approximation with rival motifs developed in the previous section.

\begin{figure}[t]
\centering
\begin{minipage}{0.32\linewidth}
  \centering
  \pdftooltip{\includegraphics[width=\linewidth]{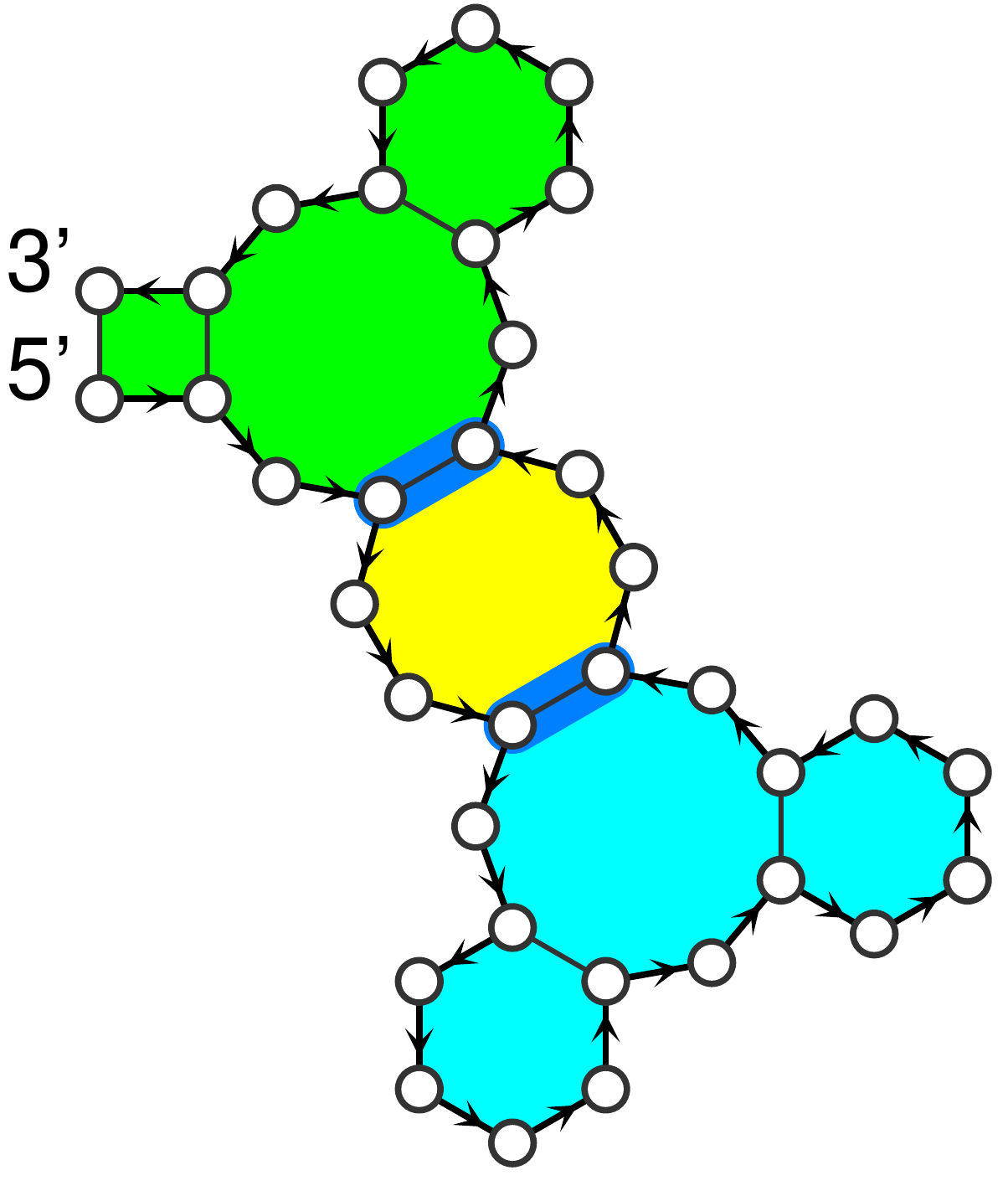}}{First decomposition of the multilooping fun structure into color-highlighted non-overlapping motifs.}
\end{minipage}\hfill
\begin{minipage}{0.32\linewidth}
  \centering
  \pdftooltip{\includegraphics[width=\linewidth]{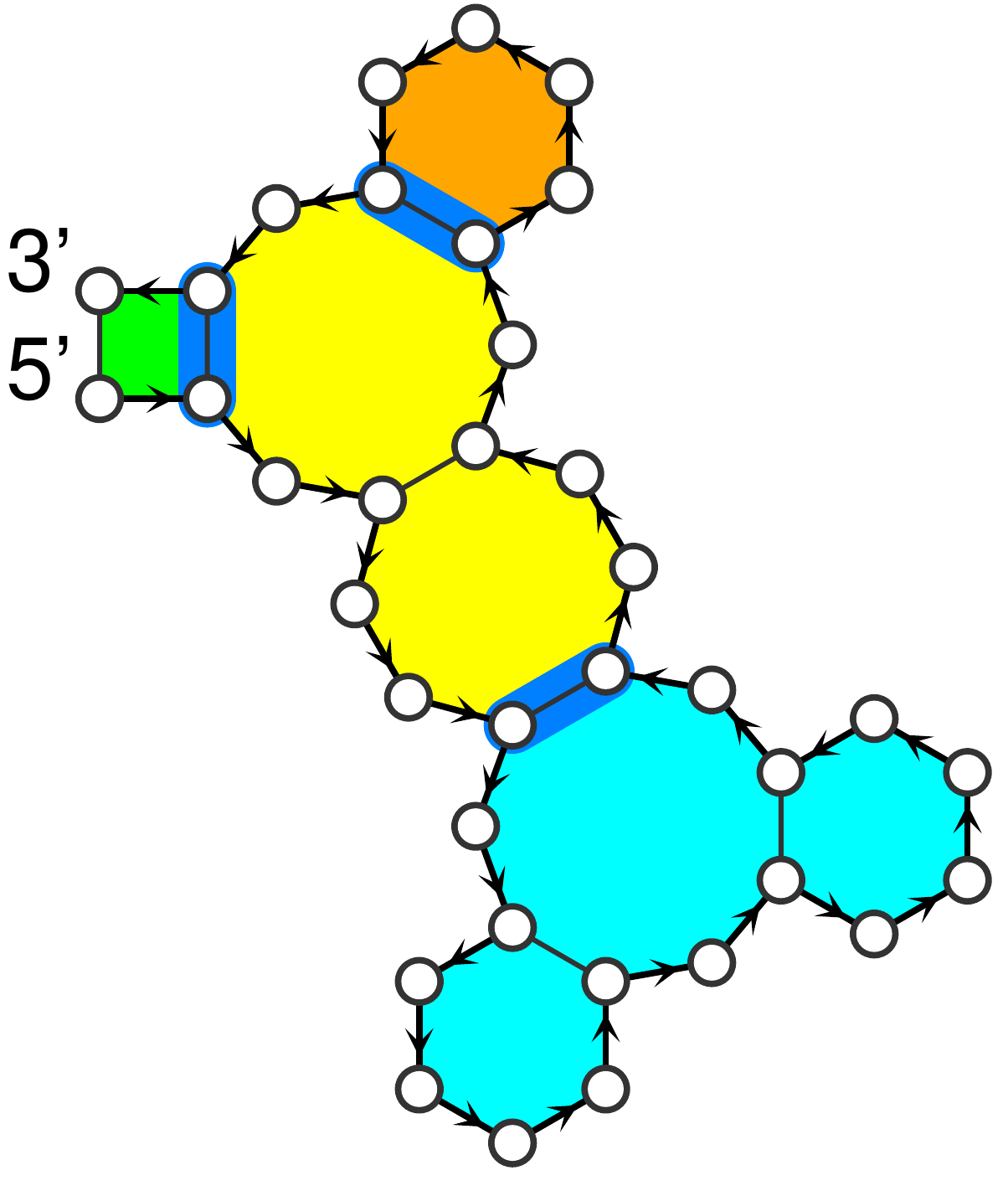}}{Second decomposition of the multilooping fun structure into a different set of color-highlighted non-overlapping motifs.}
\end{minipage}
\begin{minipage}{0.32\linewidth}
  \centering
  \pdftooltip{\includegraphics[width=\linewidth]{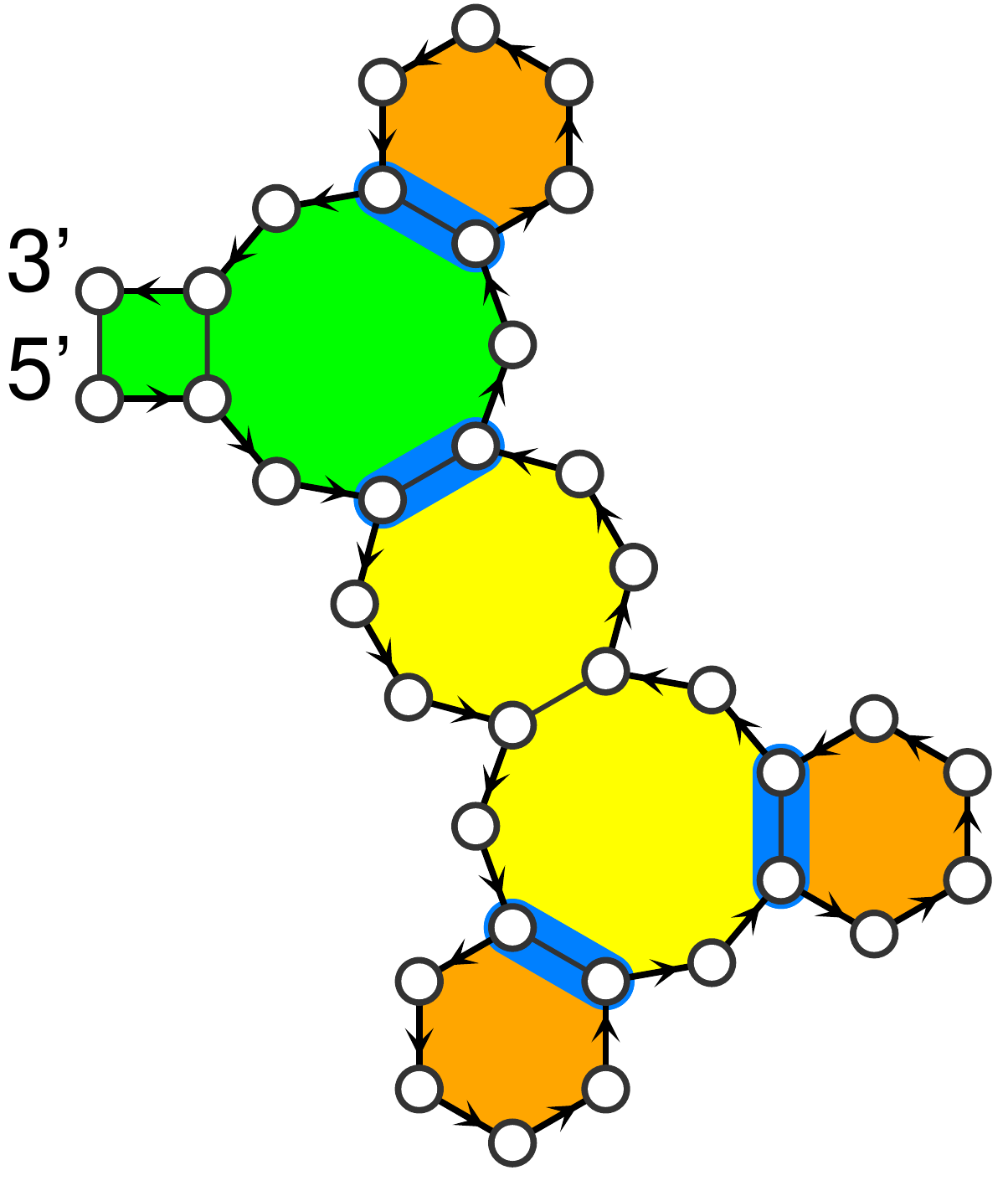}}{Third decomposition of the multilooping fun structure into another set of color-highlighted non-overlapping motifs.}
\end{minipage}
\caption{Three different decompositions for the same structure for the same structure shown in Fig.~\ref{fig:decomposition_eterna57}. \finaledit{Within each decomposition, different motifs are highlighted in different colors.}}
\label{fig:threedecompositions}
\end{figure}


Figure~\ref{fig:threedecompositions} shows three alternative decompositions of the same Eterna100 structure from Fig.~\ref{fig:decomposition_eterna57}.
In fact, a given  structure  admits exponentially many decompositions, because each base pair may or may not be selected as a splitting boundary. This yields the following \finaledit{Theorem}.

\begin{theorem}
\label{thm:total_decompose}
The total number of distinct decompositions of a secondary structure $\vecy$ is $2^{|\pairs(\vecy)|}$.
\end{theorem}

Let $\mathcal{M}$ denote the set of decompositions that can be evaluated efficiently for a target structure $\ystar$. The best probability bound obtainable via decomposition is then
\begin{equation}
\pbound(\ystar) \leq \min_{\boldsymbol{M} \in \mathcal{M}}  \prod_{\boldsymbol{m} \in \boldsymbol{M}} \pbound(\boldsymbol{m}).\label{eq:min_product}
\end{equation}

\vspace{-0.6cm}
\subsection{Optimal Decomposition Search by Linear-Time Dynamic Programming }

\begin{figure}[t]
\begin{minipage}[c]{0.52\linewidth}
  \centering
  \pdftooltip{\includegraphics[width=\linewidth]{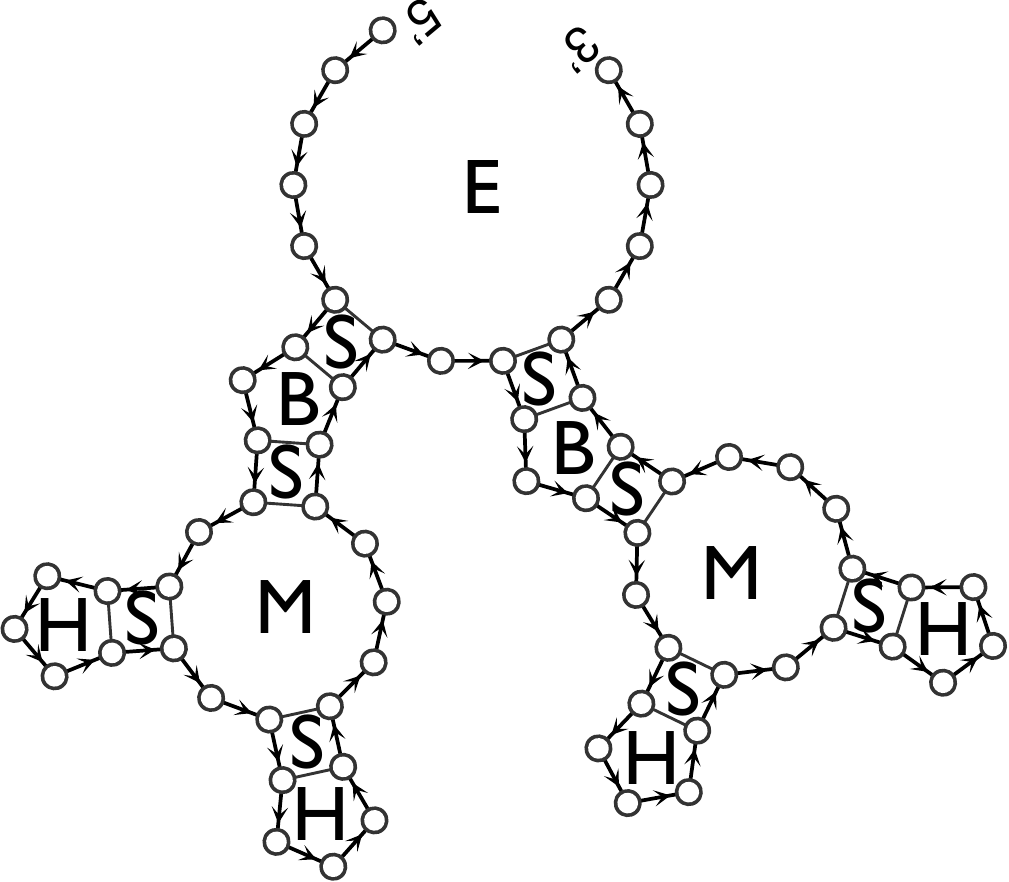}}{Diagram of the Eterna100 RNA secondary structure named Chicken Feet, showing its loop tree pat and characteristic symmetric branching pattern.}
\end{minipage}
\; 
\begin{minipage}[c]{0.45\linewidth}
\begin{lrbox}{\looptreebox}
\scalebox{0.7}{
\begin{tikzpicture}[
  >=Latex,
  node distance=2.5mm and 4.5mm,
  loop/.style={circle, draw=black, very thick, minimum size=2.5mm, inner sep=0pt,
               font=\bfseries\large, align=center},
  edge/.style={->, line width=1.1pt}
]
\node[loop] (E) {E};

\node[loop, below left=3mm and 7mm of E] (LS1) {S};
\node[loop, below=3.5mm of LS1]            (LB)  {B};
\node[loop, below=3.5mm of LB]             (LS2) {S};
\node[loop, below=3.5mm of LS2]            (LM)  {M};

\node[loop, below right=3mm and 7mm of E] (RS1) {S};
\node[loop, below=3.5mm of RS1]              (RB)  {B};
\node[loop, below=3.5mm of RB]               (RS2) {S};
\node[loop, below=3.5mm of RS2]              (RM)  {M};

\node[loop, below left=3.5mm and 3.5mm of LM] (LMS1) {S};
\node[loop, below right=3.5mm and 3.5mm of LM](LMS2) {S};
\node[loop, below=3.5mm of LMS1]             (LH1)  {H};
\node[loop, below=3.5mm of LMS2]             (LH2)  {H};

\node[loop, below left=3.5mm and 3.5mm of RM] (RMS1) {S};
\node[loop, below right=3.5mm and 3.5mm of RM](RMS2) {S};
\node[loop, below=3.5mm of RMS1]             (RH1)  {H};
\node[loop, below=3.5mm of RMS2]             (RH2)  {H};

\draw[edge] (E) -- (LS1);
\draw[edge] (E) -- (RS1);

\draw[edge] (LS1) -- (LB);
\draw[edge] (LB)  -- (LS2);
\draw[edge] (LS2) -- (LM);

\draw[edge] (RS1) -- (RB);
\draw[edge] (RB)  -- (RS2);
\draw[edge] (RS2) -- (RM);

\draw[edge] (LM) -- (LMS1);
\draw[edge] (LM) -- (LMS2);
\draw[edge] (LMS1) -- (LH1);
\draw[edge] (LMS2) -- (LH2);

\draw[edge] (RM) -- (RMS1);
\draw[edge] (RM) -- (RMS2);
\draw[edge] (RMS1) -- (RH1);
\draw[edge] (RMS2) -- (RH2);
\end{tikzpicture}
}
\end{lrbox}
\pdftooltip{\usebox{\looptreebox}}{Loop tree diagram for the Chicken Feet structure, rooted at an external loop (E), with two symmetric branches each containing a stack (S), bulge (B), stack (S), and multiloop (M), where each multiloop further splits into two paths each ending in a stack (S) and hairpin (H).}
\end{minipage}
\caption{Left: the Eterna100 structure \emph{Chicken Feet}. Right: corresponding loop tree where each loop (including stacking) is represented as a node. E, S, B, H, M represent different loop types.}
\label{fig:loop_tree}
\end{figure}

\begin{figure}[b]
\vspace{-0.5cm}
  \centering
  \pdftooltip{\includegraphics[width=\linewidth]{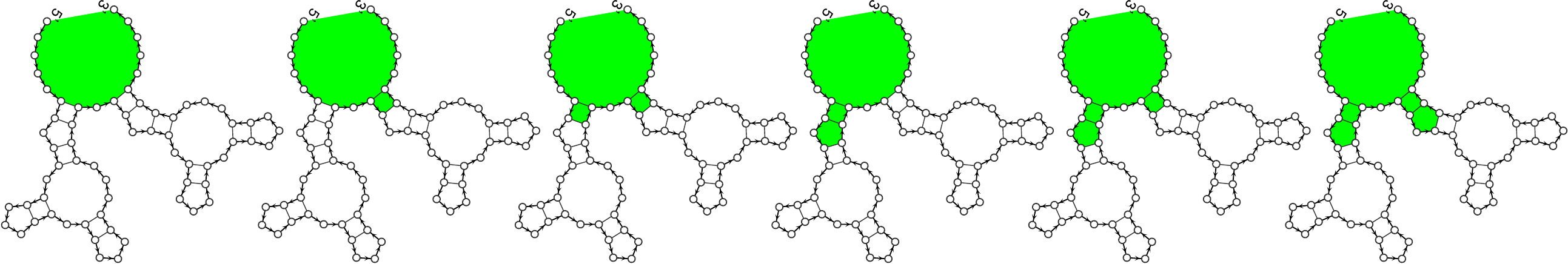}}{Grid of examples showing motif candidates generated at the root (external) loop node, with varying depth, width, and number of loops, where each candidate motif is highlighted on the RNA structure.}
  \caption{Examples of motif candidates generated at the root node. For demonstration, here the maximum values for \emph{depth}, \emph{width}, and \emph{number of loops} are set as 3, 2, and 5, respectively, and 6 out of 9 cases are shown. In  experiments the maximum \emph{depth} is set as 5.}\
  \label{fig:motif_cands}
\end{figure}


Although enumerating all decompositions is intractable, exploring more decompositions can only tighten the resulting upper bound. The key question is therefore how to search a large space of decompositions efficiently.


Inspired by dynamic programming techniques from compilers~\cite{aho1972theory,aho1976optimal} and machine translation~\cite{huang2006statistical}, we propose a linear-time memoized top-down algorithm (Algorithm~\ref{alg:topdown-dp}) that explores many decompositions while remaining computationally efficient.
We represent a target structure by a fixed loop tree $\tau$, where each node corresponds to a loop in the structure (Fig.~\ref{fig:loop_tree}).
We decompose the structure by traversing $\tau$ top-down, starting at the root. 
When visiting a loop node $\eta$, we are solving the subproblem for the subtree $\tau_{\eta}$, and we must choose a motif $\vecm_\eta$ that contains $\eta$ 
\finaledit{and some of its descendants}.
To avoid enumerating all possibilities, we constrain candidate motifs using three parameters: \emph{depth}, \emph{width}, and \emph{number of loops}.
Figure~\ref{fig:motif_cands} shows examples of motif candidates generated when visiting the external loop at the root; the candidate-generation procedure is described in Algorithm~\ref{alg:motif_gen} (in Supplementary Section~\ref{sec:alg}).
Once $\vecm_\eta$ is selected, we recursively process each adjacent descendant loop node $\eta_i$ in the  subtree $\tau_\eta$. 
This procedure yields the following dynamic programming recurrence:
\begin{equation}
best(\tau_\eta) = \hspace{-0.5cm}\min_{\vecm_\eta \in \Call{MotifGen}{\eta}} \pbound(\vecm_\eta) \hspace{-0.7cm}\prod_{\eta_i \in \mathit{descendents}(\vecm_\eta)} \hspace{-0.5cm}best(\tau_{\eta_i}),\label{eq:recurrence}
\end{equation}
where $best(\tau_\eta)$ represents the tightest upper bound for the substructure rooted at the loop node $\eta$.
The base case occurs when $\eta$ is a hairpin loop, for which $best(\tau_\eta)=1$.

Each loop node may admit many candidate motifs; the number of possible decompositions can also be calculated by dynamic programming recurrence:
\begin{equation}
count(\tau_\eta) = \hspace{-0.5cm}\sum_{\vecm_\eta \in \Call{MotifGen}{\eta}} \prod_{\eta_i \in \mathit{descendents}(\vecm_\eta)} \hspace{-0.5cm}count(\tau_{\eta_i}),\label{eq:count}
\end{equation}
with the base case value of 1 \finaledit{(leaf nodes)}.
As a result, the decomposition possibilities grows exponentially with the size of the loop tree. 
We address this via memoization: each subtree rooted at a loop node is solved at most once and cached.
We also store backpointers to recover the best motif choice at each node, and reconstruct the optimal decomposition by backtracking.
The resulting dynamic program runs in $O(\alpha\beta)$ time, where $\alpha$ is the number of loop nodes in the loop tree and $\beta$ is the maximum number of motif candidates considered at any node.

\begin{algorithm}[t]
\caption{Top-down Memoized Decomposition}
\label{alg:topdown-dp}
\begin{algorithmic}[1]
\Function{Decompose}{$\eta$}
  \If{$cache[\eta]$ defined}
    \State \Return $cache[\eta]$
  \EndIf
  \State $best \gets 1$
  \State $candidates \gets \Call{MotifGen}{\eta, depth, width, max\_loop}$
  \For{$\vecm_{\eta} \in candidates$} \Comment{try each candidate $\vecm_{\eta}$}
      \State $product \gets \mathit{\pbound}(\vecm_{\eta})$\Comment{exact or approx.}\label{line:mprob}
      \State $descendants \gets $ {loop nodes adjacent to \vecm} in $\tau_{\eta}$
      \For{$\eta_i \in descendants$}
        \State $(p_i, \vecm_{\eta_i}^{\text{best}}) \gets \Call{Decompose}{\eta_i}$ \Comment{recursion}
        \State $product \gets product \cdot p_i$
      \EndFor
      \If{$product < best$}
        \State $best \gets product$
        \State $\vecm_{\eta}^{\text{best}} \gets \vecm_{\eta}$ \Comment{plug in the results}
      \EndIf
  \EndFor
  \State $cache[\eta] \gets (best, \vecm_{\eta}^{\text{best}})$ \Comment{memoization}
  \State \Return $cache[\eta]$ \Comment{optimal decomposition choice}
\EndFunction
\end{algorithmic}
\end{algorithm}

\vspace{-0.3cm}

\section{Evaluation Results}\label{sec5}

\subsection{Implementation}
\finaledit{Our work}, \emph{\finaledit{LinearDecompose}}, consists of two core components:
(i) {Linear-Time Dynamic Decomposition} (Alg.~\ref{alg:topdown-dp}) and
(ii) {Ensemble Approximation} (Alg.~\ref{alg:approxe_2}).
The \finaledit{first} 
module efficiently searches for the optimal motif decomposition by the recurrence formula Eq.~\ref{eq:recurrence} over exponentially many candidates.
The \finaledit{second} 
module estimates probability bounds $\pbound(\ystar)$ or $\pbound(\mstar)$ for a target structure or motif.

\finaledit{Our C++ code runs} on Linux machines with 4.0\,GHz CPUs and 32\,GB memory.
It calls ViennaRNA (v2.7) folding energy
engine~\cite{lorenz+:2011} \finaledit{with Turner 2004 parameters~\cite{Mathews+:2004}}.
Ensemble approximation is parallelized with OpenMP over 90  cores.

During decomposition search, Algorithm~\ref{alg:topdown-dp} repeatedly queries motif probability bounds $\pbound(\vecm_\eta)$ (line~\ref{line:mprob}).
To improve efficiency, we precompute and cache motif bounds using a hybrid strategy:
(1) exact bounds via exhaustive enumeration for motifs of length \finaledit{up to 14}, 
following CountingDesign~\cite{yao:2021thesis}; and
(2) non-exact bounds via ensemble approximation for larger motifs.

\subsection{Datasets and Baselines}
We evaluate 
on two popular RNA structure datasets:
\begin{enumerate}
\item \finaledit{\textbf{ArchiveII}}~\citep{sloma+mathews:2016,Cannone+:2002} contains native RNA secondary structures spanning 10 families of naturally  RNAs, including tRNA, rRNA, etc. We remove pseudoknotted structures and those incompatible with ViennaRNA's default loop constraints. We \finaledit{only use} structures \finaledit{longer than $200$\,{\em nt}}, resulting in 1,144 structures.
\item \textbf{Eterna100}~\citep{anderson+:2016} consists of 100 secondary structures \finaledit{of varying design difficulty} designed by human players of Eterna and is widely used for evaluating RNA design and designability.
\end{enumerate}
We compare \finaledit{our method} against CountingDesign and \finaledit{our enhanced version of it},  \finaledit{CountingDesign+}:
\begin{enumerate}
\item \textbf{CountingDesign}~\cite{yao+:2019,yao:2021thesis} can identify exact probability bounds for very short motifs (up to 14 nucleotides) by exhaustively enumerating RNA sequences for each motif. However, the method is not scalable and lacks explainability.
\item \textbf{CountingDesign+}. We extend CountingDesign to include motifs with external loops, 
\finaledit{thus expanding their  motif definition to align with ours}.
\end{enumerate}

\vspace{-0.3cm}

\subsection{Overall \finaledit{and Individual Results}}
\begin{table}[t]
\caption{Average probability bounds on each dataset.}
\label{tab:overall}
\resizebox{\columnwidth}{!}{%
\begin{tabular}{l|lrr}
\hline
\multicolumn{2}{l}{Dataset}                   & ArchiveII & Eterna100 \\ \hline
\multicolumn{2}{l}{\# of unique structures}      & 1144     & 100     \\ 
\multicolumn{2}{l}{best achieved $p(\ystar | \vecx)\finaledit{\uparrow}$}  
  & 0.766     & 0.599     \\ \hline
\midrule
\multirow{5}{*}{\rotatebox{90}{\finaledit{$\pbound(\ystar)\finaledit{\downarrow}$}}} &
CountingDesign     &  \finaledit{0.896\tiny$\pm 0.242$}           & \finaledit{0.898\tiny$\pm 0.234$} \\ 
&CountingDesign+     & 0.875\finaledit{\tiny$\pm 0.255$}     & 0.876\finaledit{\tiny$\pm 0.250$}     \\ 
\cline{2-4} 
&\finaledit{LinearDecompose} & \textbf{0.800}\finaledit{\tiny$\pm 0.288$}    & \textbf{0.741}\finaledit{\tiny$\pm 0.318$}     \\ 
&\editcolor no ensemble approx.\hspace{-1.5cm} & \editcolor 0.866\tiny$\pm 0.256$ & \editcolor 0.839\tiny$\pm 0.275$ \\
&\editcolor no dynamic decomp.\hspace{-1.5cm}  & \editcolor 0.819\tiny$\pm 0.286$ & \editcolor 0.787\tiny$\pm 0.295$ \\ \hline
\end{tabular}}
\vspace{-0.5cm}
\end{table}

Table~\ref{tab:overall} reports average probability bounds on both datasets.
For reference, we also include the highest $p(\ystar|\vecx)$ achieved by \finaledit{the state-of-the-art RNA design methods, taking the better of SAMFEO~\cite{zhou+:2023samfeo} and SamplingDesign~\cite{tang+:2026} for each puzzle.}

{\finaledit{LinearDecompose}} consistently yields the tightest bounds, improving over CountingDesign+ by 0.075 on ArchiveII and 0.135 on Eterna100.
This improvement is expected, as {\finaledit{LinearDecompose}} is theoretically guaranteed to produce bounds no worse than CountingDesign+.

We also observe a clear correlation between probability bounds and empirical designability.
ArchiveII structures are generally more designable: the average bound (0.800) is close to the achieved probability (0.766).
In contrast, Eterna100 structures are less designable, with a larger gap between the achieved probability (0.599) and the bound (0.741).
\finaledit{This difference is statistically significant (Mann-Whitney U\,=\,72{,}687, $p = 6.95 \times 10^{-6}$, $r = 0.27$).}

\begin{figure}[t]
  \centering
  \pdftooltip{%
  \includegraphics[width=\linewidth]{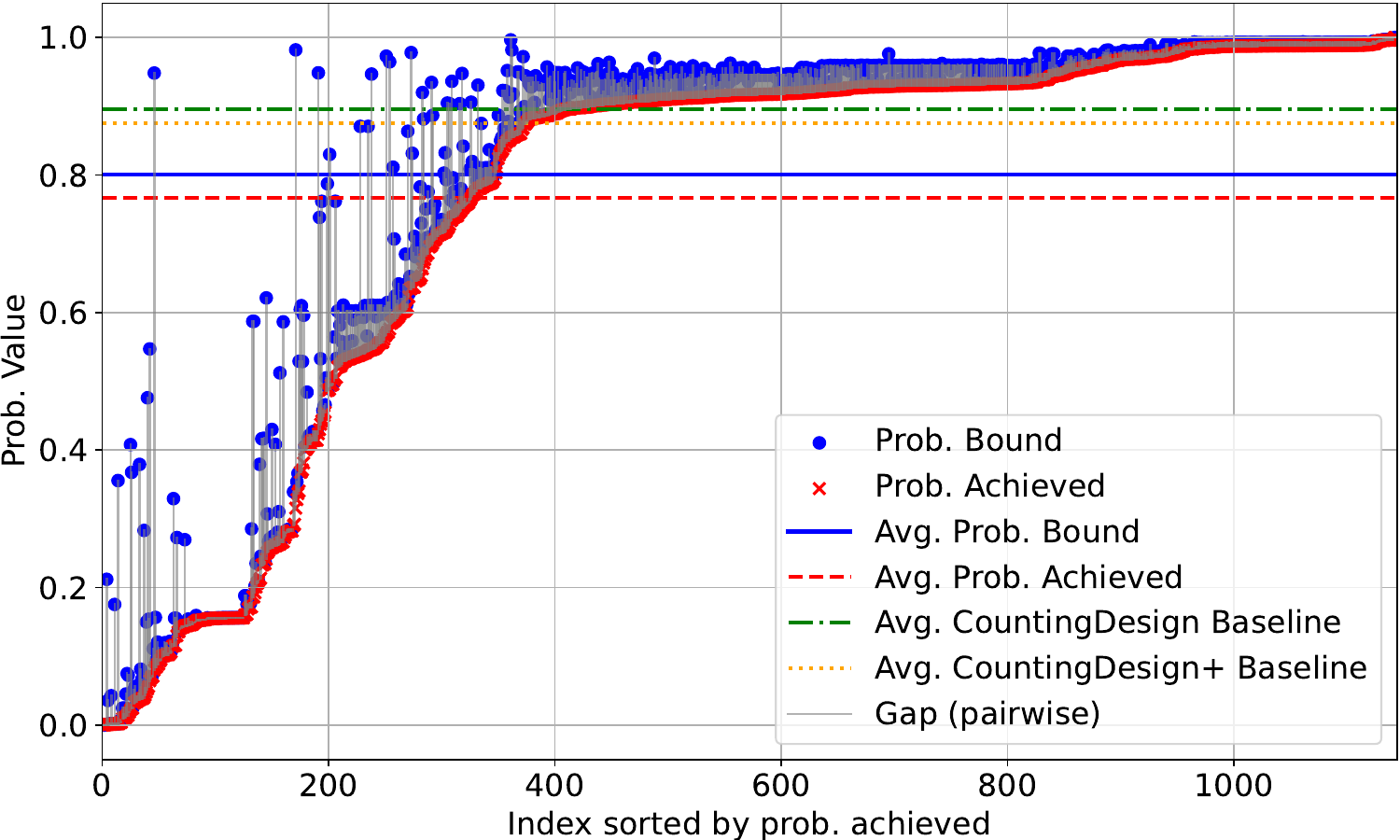}%
}{Scatter plot of probability bounds vs.\ achieved probabilities for 1144 ArchiveII structures. LinearDecompose yields tighter bounds than baselines.}
  \caption{Probability bounds vs.~achieved $p(\ystar\mid\vecx)$ on ArchiveII.}\
  \label{fig:bounds_archiveii}
%
  \centering
  \pdftooltip{%
  \includegraphics[width=\linewidth]{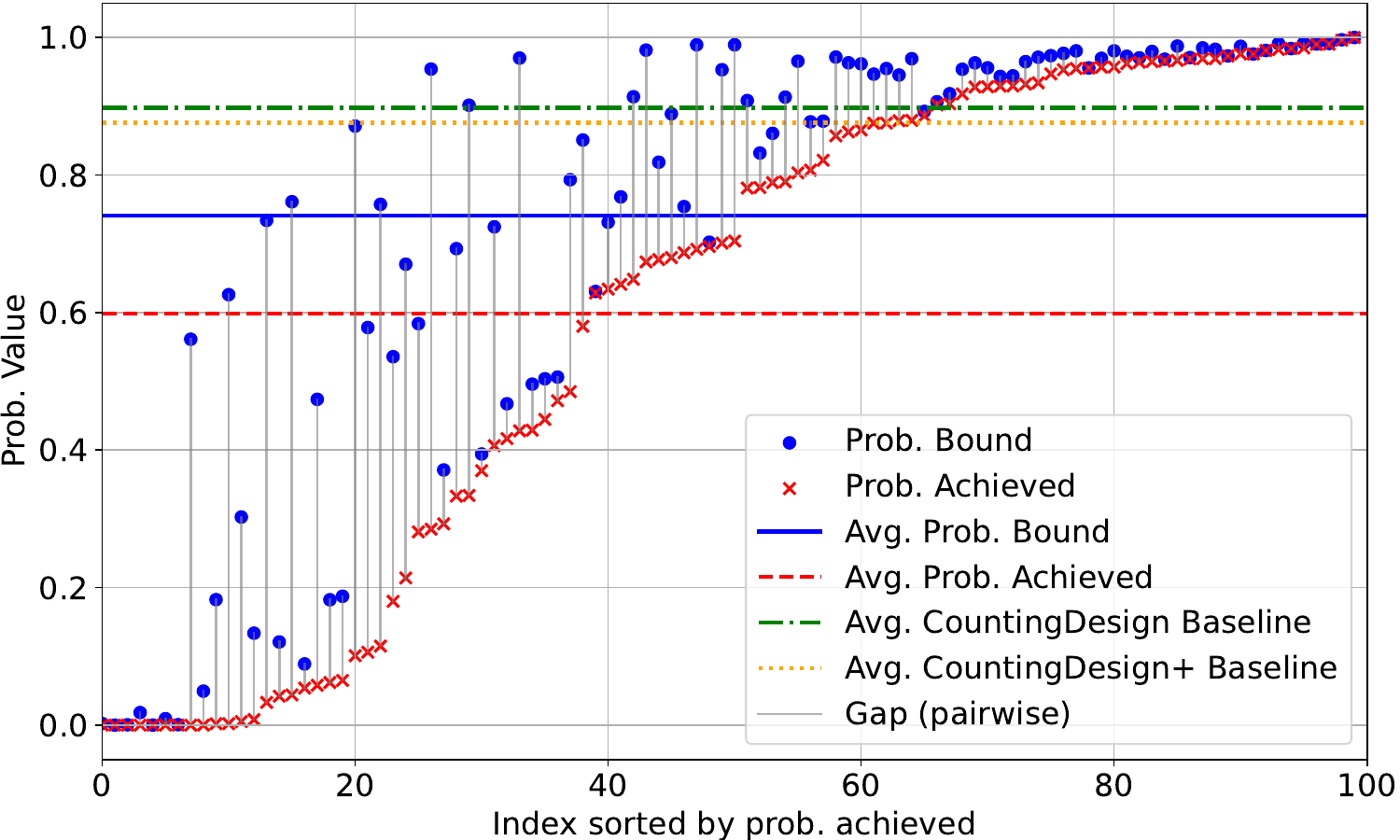}%
}{Scatter plot of probability bounds vs.\ achieved probabilities for 100 Eterna100 structures. LinearDecompose yields tighter bounds than baselines, with a larger gap between bounds and achieved probabilities compared to ArchiveII.}
  \caption{Probability bounds vs.~achieved $p(\ystar\mid\vecx)$ on Eterna100.}\
  \label{fig:bounds_eterna100}
  \vspace{-0.8cm}
\end{figure}

Figs.~\ref{fig:bounds_archiveii}--\ref{fig:bounds_eterna100} show per-structure bounds and best achieved $p(\ystar\,|\,\vecx)$.
Bounds tend to be tighter for both highly designable and highly undesignable structures.
Highly designable structures typically consist of motifs with uniformly high local probabilities, while highly undesignable structures often contain multiple motifs with very low bounds, \finaledit{often} leading to near-zero global bounds.

These bounds also provide a useful lens for evaluating RNA design results.
For example, a structure achieving probability 0.7 may be closer to its theoretical limit than another achieving 0.8 if its bound is substantially lower.

Table~\ref{tab:plots} presents a structural dissection of several Eterna100 puzzles previously shown to be \UMFE-undesignable~\cite{zhou+:2025scalable}.
For each structure $\ystar$, we report the number of decompositions explored by \finaledit{LinearDecompose}
 (via dynamic programming Eq.~\ref{eq:count}), 
\finaledit{the total number of possible decompositions (via Theorem~\ref{thm:total_decompose})},
the optimal decomposition achieving the tightest probability bound,  the resulting global probability upper bound $\pbound(\ystar)$,
\finaledit{and the best $p(\ystar\mid \vecx)$ achieved by SAMFEO and SamplingDesign (as a lowerbond)}.


\finaledit{LinearDecompose efficiently explores exponentially many decompositions per puzzle (e.g., 47M out of 67M total decompositions for the first puzzle)}, 
and consistently identifies a small set of local motifs whose combined bounds dominate the global probability bound.
Notably, structures with complex multiloop or bulge configurations (e.g., \emph{multilooping fun}) exhibit extremely low bounds, indicating that probabilistic undesignability can often be traced to a few highly constraining motifs.
\finaledit{In addition, we also noticed that motifs with isolated base pairs tend to have low probability bounds, consistent with our prior systematic study~\cite{zhou+:2026:jcb} on undesignable motifs that isolated base pairs are energetically unfavorable and often avoided in natural RNA structures.}
These results highlight the value of decomposition-based analysis for pinpointing the structural origins of low RNA designability.

\begin{table}[t]
\centering
\caption{Structure Dissection with Eterna100 Puzzles.}
\label{tab:plots}
\resizebox{1.0\columnwidth}{!}{%
\begin{tabular}{|p{1.55cm} | >{\raggedleft\arraybackslash}p{1.55cm} | c| >{\raggedleft\arraybackslash}p{1.4cm}|}
\hline
\multirow{2}{*}{Puzzle (\ystar)}
& \multicolumn{2}{c|}{Decomposition}
& \multirow{1}{*}{\!\!\small $\pbound(\ystar\!)$}\\ 
& \finaledit{Explored} & Optimal &  \!\!\finaledit{\small max $p(\ystar|\vecx)$}\!\!\\ 
\hline
\finaledit{Mat - Elements \& Sections} & \finaledit{$47,012,668$} \finaledit{out of $2^{26}=67,108,864$} & \raisebox{-.5\height}{\pdftooltip{\includegraphics[width=0.18\textwidth]{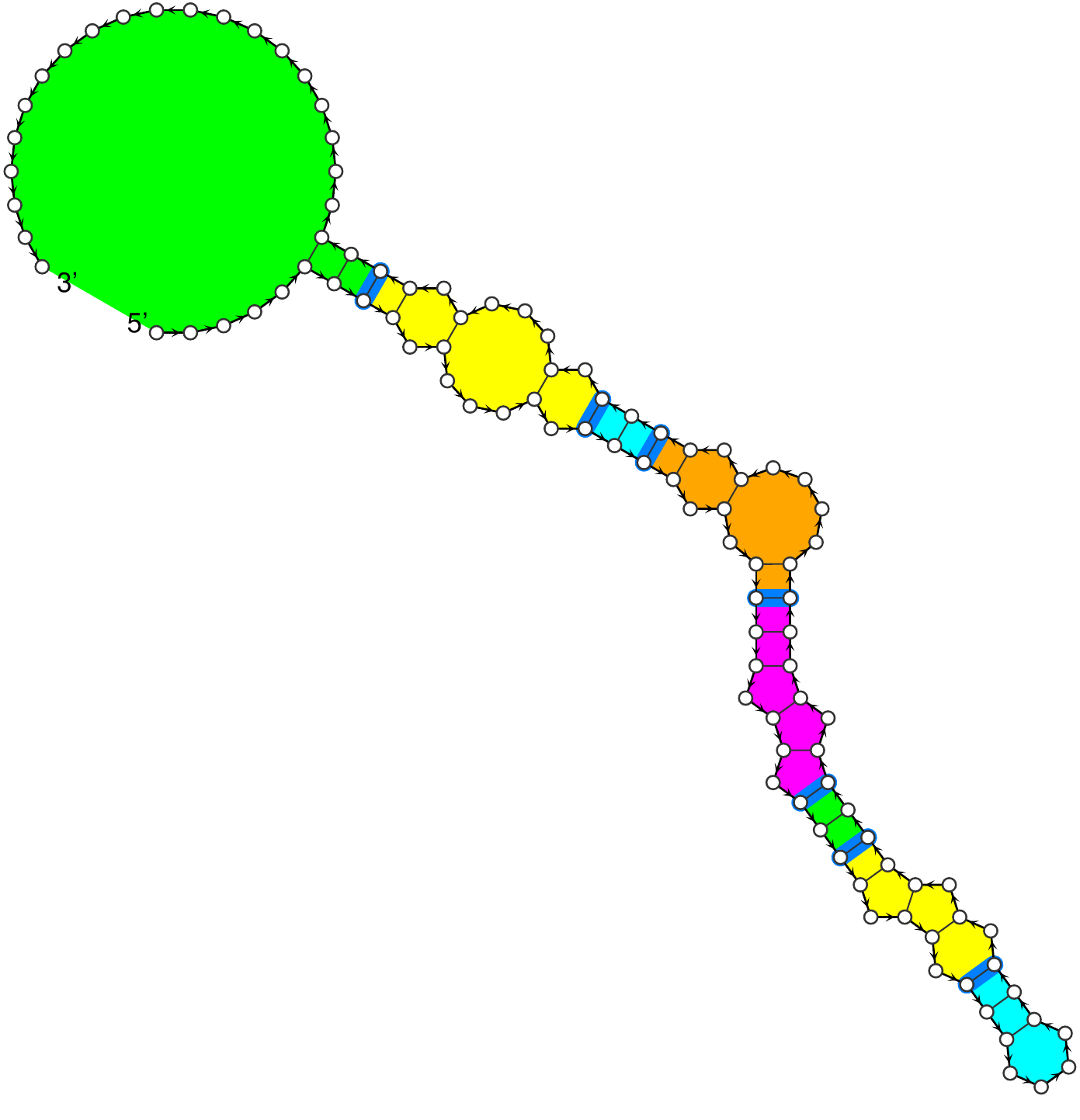}}{Optimal motif decomposition for the ``Mat - Elements \& Sections'' structure, \pbound 0.0889.}} & \finaledit{0.0889$\downarrow$}
\finaledit{0.0426$\uparrow$} \\ 
\hline
\finaledit{Loop next to a Multiloop} & \finaledit{$2,997,104$} \finaledit{out of $2^{22}$= $4,194,304$}& \raisebox{-.5\height}{\pdftooltip{\includegraphics[width=0.18\textwidth]{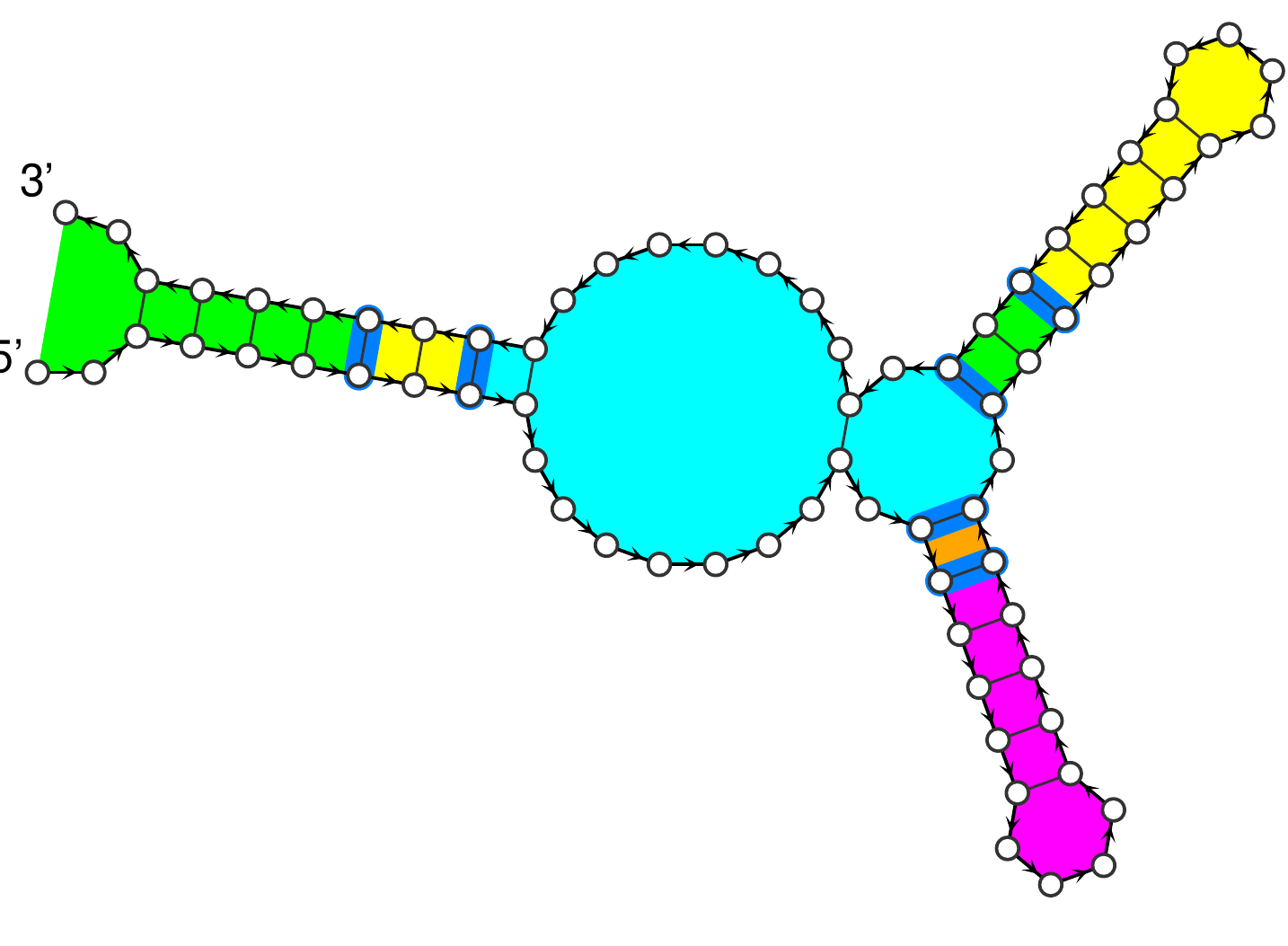}}{Optimal motif decomposition for the ``Loop next to a Multiloop'' structure, \pbound 0.4958.}} & \finaledit{0.4958$\downarrow$}
\finaledit{0.4293$\uparrow$} \\ 
\hline
\finaledit{Simple Single Bond} & \finaledit{32} \finaledit{out of $2^5=32$} & \raisebox{-.5\height}{\pdftooltip{\includegraphics[width=0.17\textwidth]{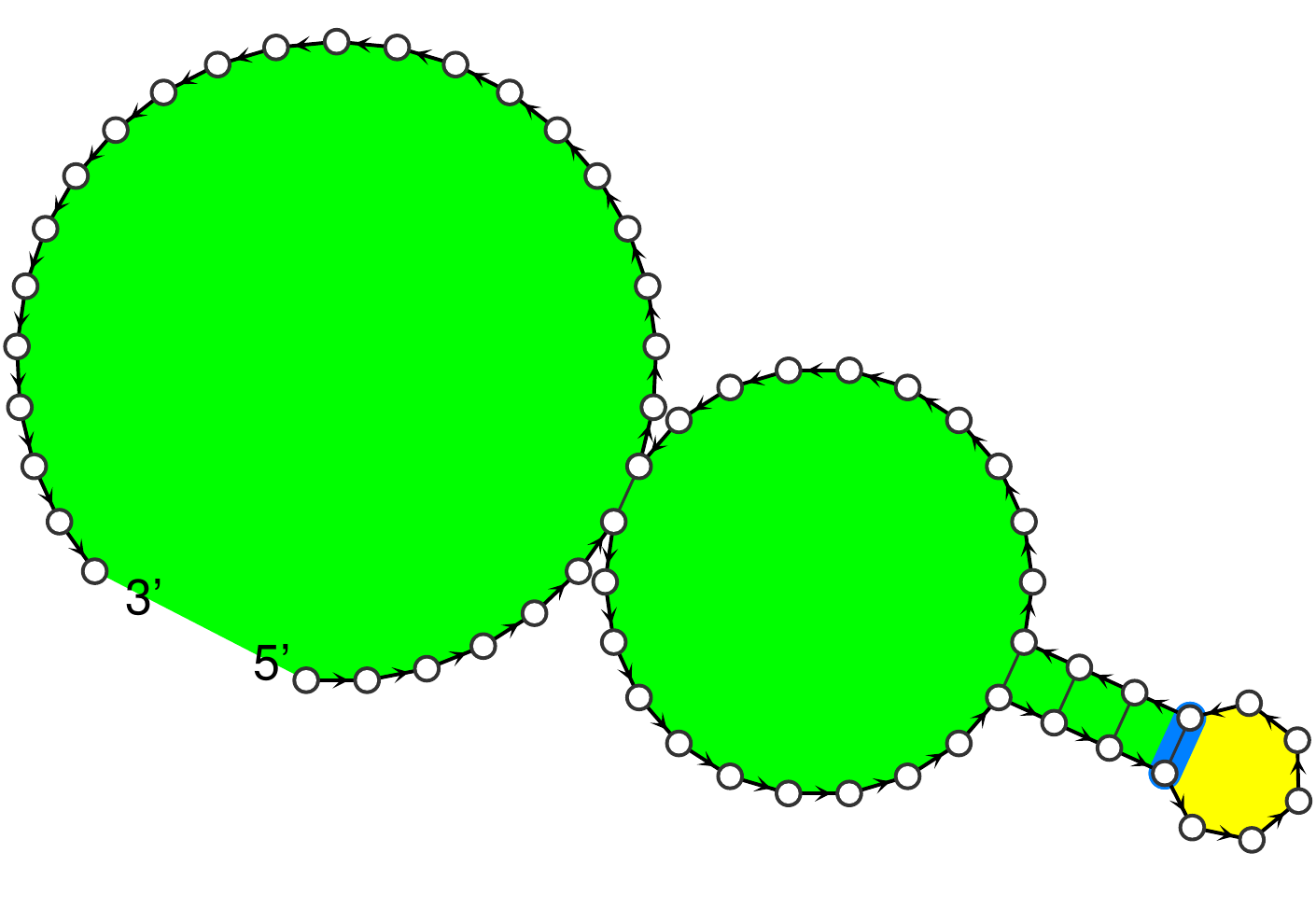}}{Optimal motif decomposition for the ``Simple Single Bond'' structure, \pbound 0.3942.}} & 
\finaledit{0.3942$\downarrow$}
\finaledit{0.3700$\uparrow$} \\ 
\hline
1, 2, 3 and 4 bulges & $2,734,333$ \finaledit{out of $2^{22}$=} \finaledit{$4,194,304$} & \raisebox{-.5\height}{\pdftooltip{\includegraphics[width=0.18\textwidth]{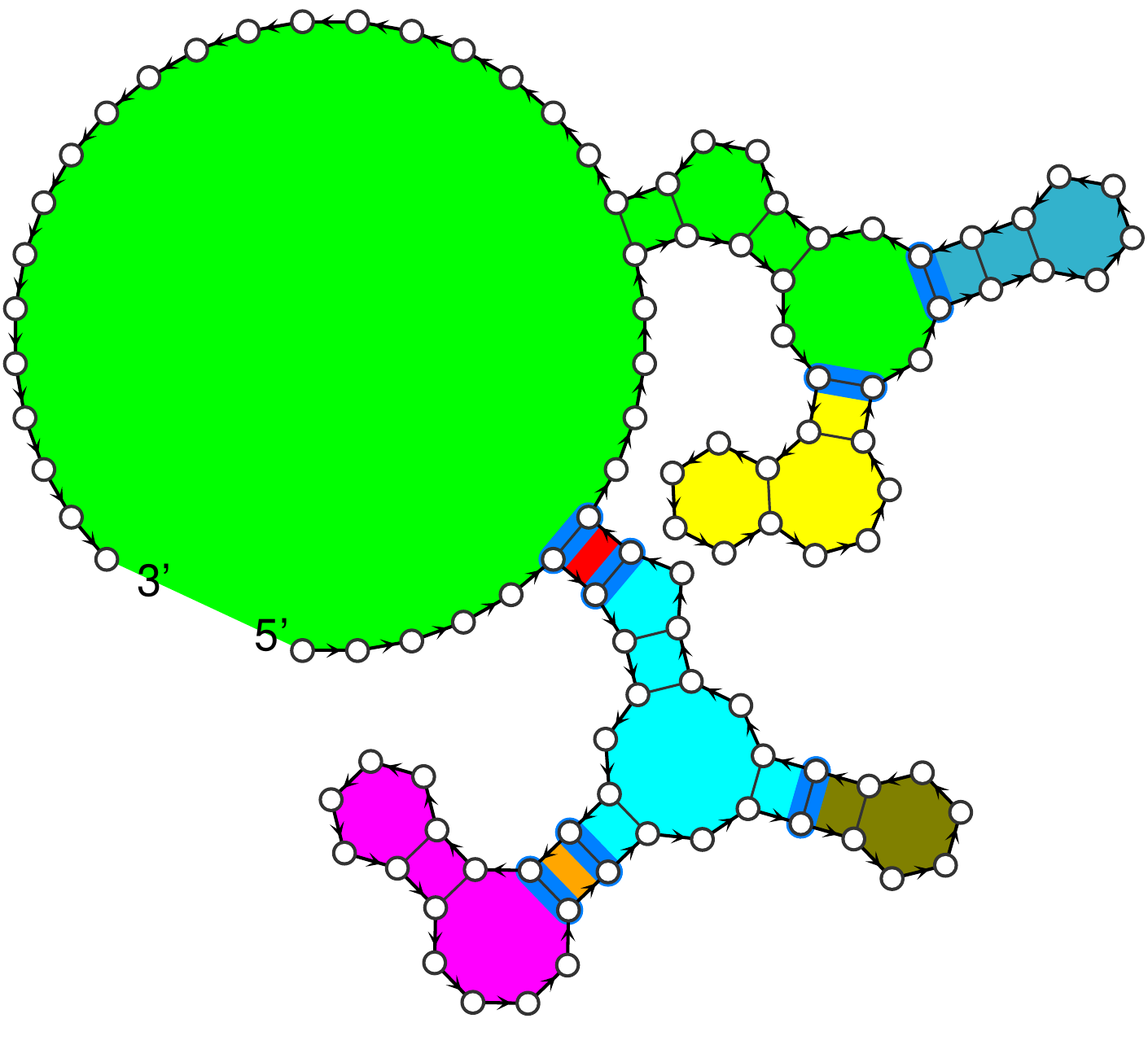}}{Optimal motif decomposition for the ''1-2-3-4 bulges'' structure, \pbound 0.3026.}} & 0.3026$\downarrow$ 
\finaledit{0.0054$\uparrow$} \\ \hline
Repetit. Seqs. 8/10 & $860$ \finaledit{out of $2^{10}\!=\!1,024$}& \raisebox{-.5\height}{\pdftooltip{\includegraphics[width=0.13\textwidth]{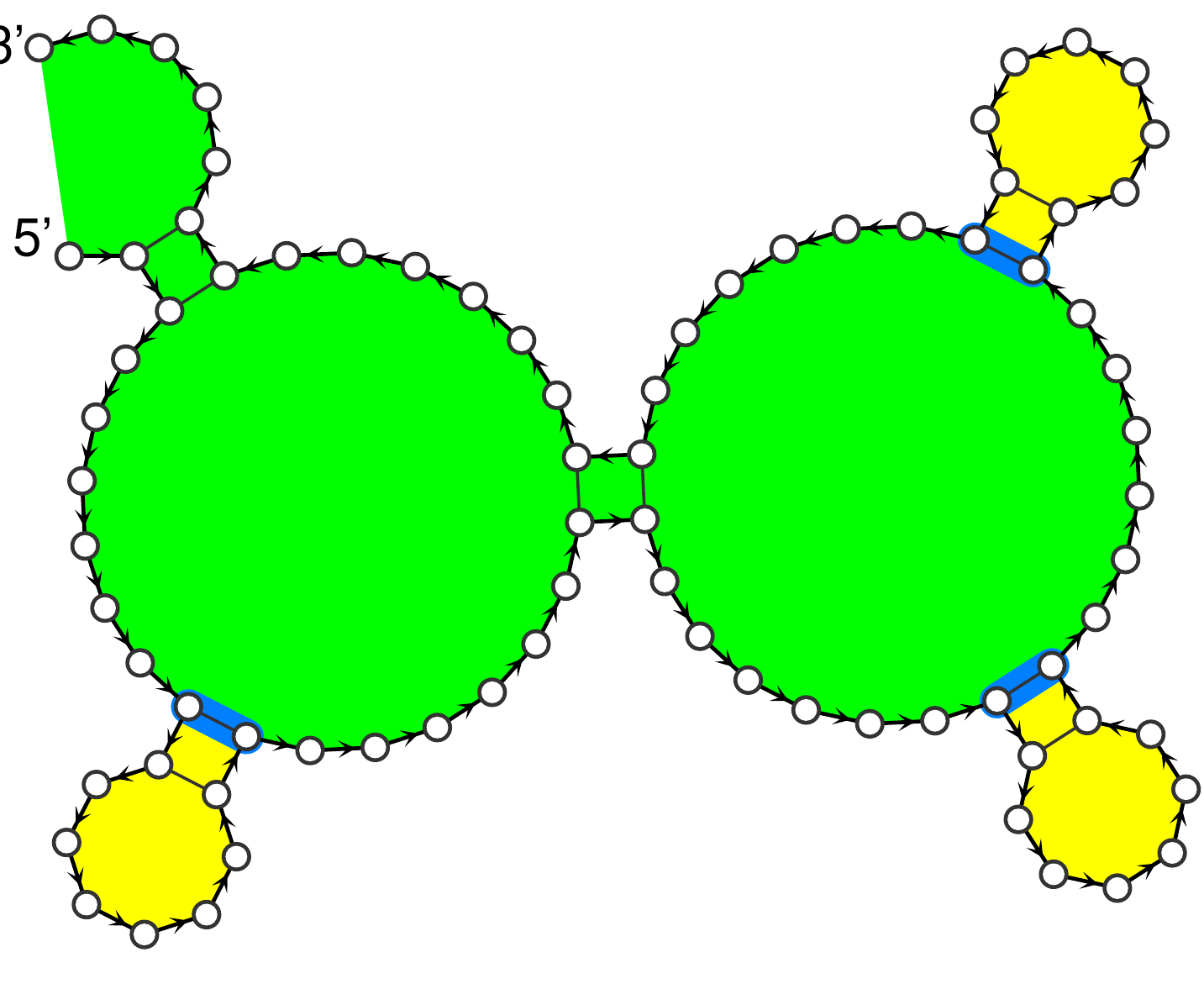}}{Optimal motif decomposition for the ''Repetitive Sequences 8/10'' structure, \pbound 0.1874.}} & 0.1874$\downarrow$ 
  \finaledit{0.0651$\uparrow$} \\ \hline
multilooping fun & $119$ \finaledit{out of $2^7=128$} & \raisebox{-.5\height}{\pdftooltip{\includegraphics[width=0.13\textwidth]{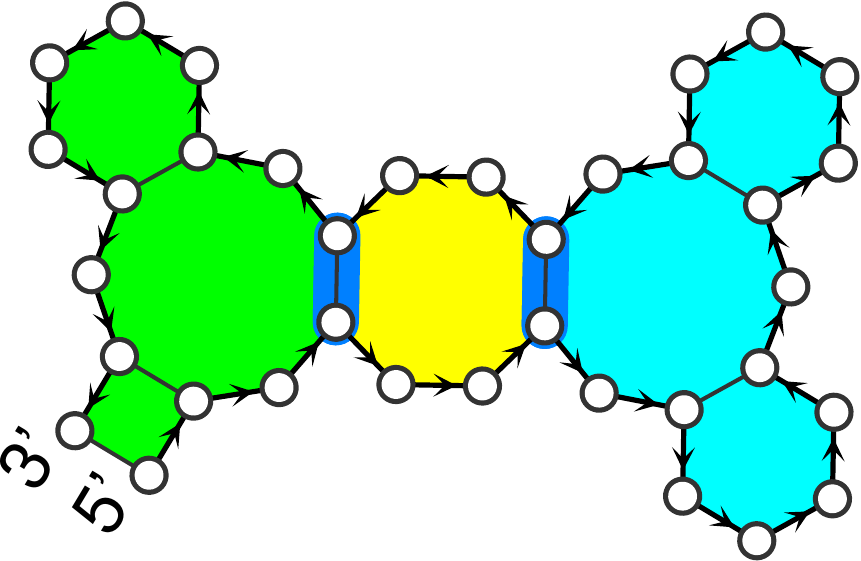}}{Optimal motif decomposition for the ''multilooping fun'' structure, \pbound 0.0006.}} & 0.0006$\downarrow$ 
  \finaledit{$2\!\times\! 10^{-6}$$\uparrow$} \\ \hline
Chicken feet & $52,114$ \finaledit{out of $2^{16}\!=\!65,536$}& \raisebox{1cm}{\pdftooltip{\includegraphics[width=0.15\textwidth,angle=-90]{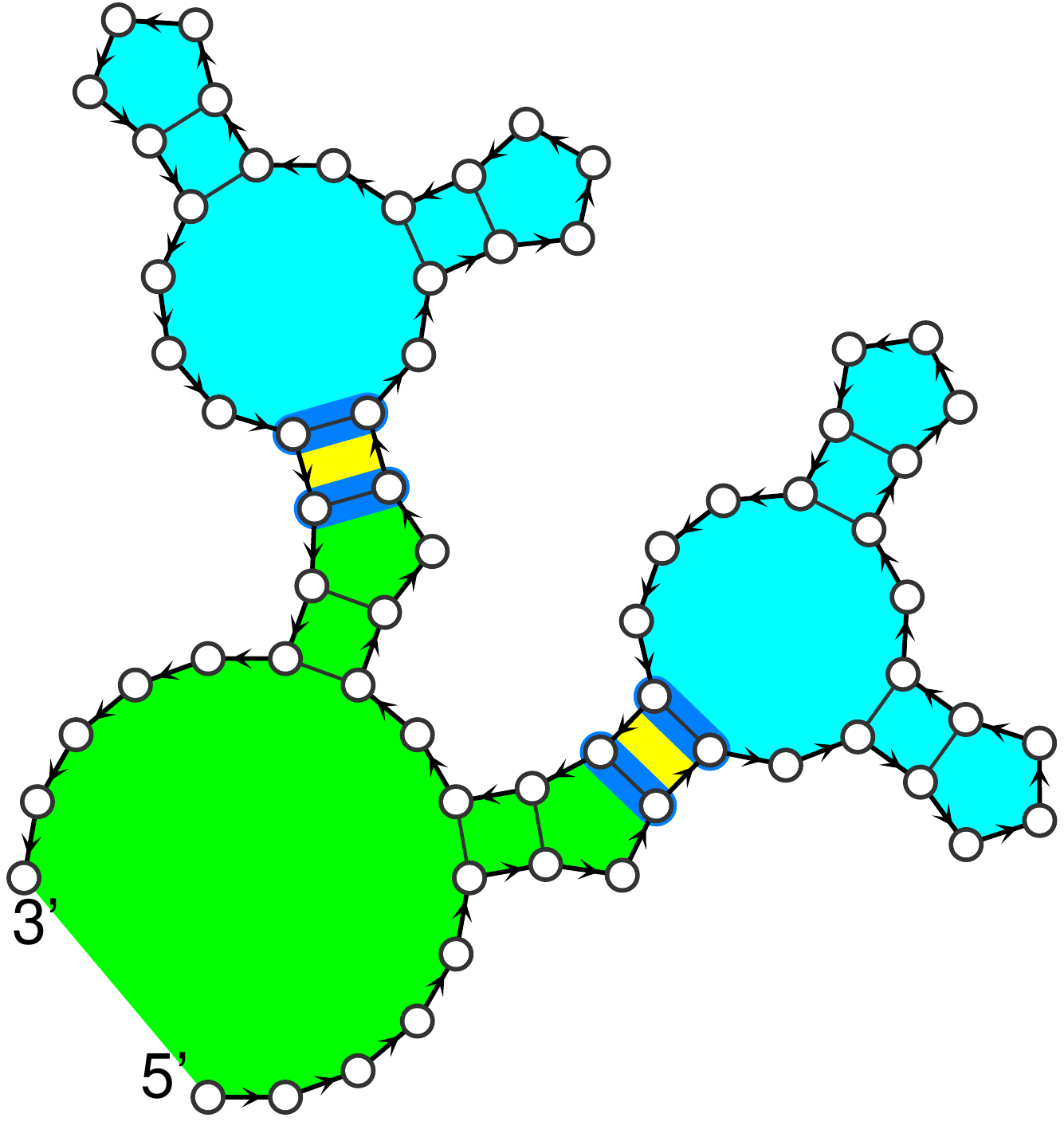}}{Optimal motif decomposition for the ''Chicken Feet'' structure, \pbound 0.1336.}} & 0.1336$\downarrow$ 
  \finaledit{0.0080$\uparrow$} \\ 
\hline
\end{tabular}
}
\end{table}

\subsection{Ablation Studies}
To assess the two core components of {\finaledit{LinearDecompose}}, we perform two ablations:
(1) disable motif ensemble approximation by setting approximate motif bounds to 1; and
(2) disable dynamic decomposition by using only the \finaledit{best} single-motif bound.
As shown in \finaledit{the last two rows of Table~\ref{tab:overall}}, both ablations lead to looser bounds, confirming the importance of both components.
Nevertheless, the ablated variants still outperform the baselines in Table~\ref{tab:overall}, demonstrating the robustness of the framework.


\subsection{Efficiency}
The runtime of \finaledit{LinearDecompose} consists of decomposition search and motif ensemble approximation.
Decomposition search is fast, averaging 0.01\,s per structure on ArchiveII and 0.9\,s on Eterna100.
Ensemble approximation estimates bounds for \finaledit{15{,}800} and \finaledit{8{,}000} unique motifs on ArchiveII and Eterna100, with an average cost of 7.6\,s and 10.7\,s per motif, respectively.
\finaledit{Approximately 10\% and 2\% of motifs on ArchiveII and Eterna100, respectively, are skipped due to exceeding the enumeration time limit, for which the bound defaults to 1.
As these fractions are small, we expect the impact on overall bound quality to be minimal.}
\finaledit{The average time costs for ensemble approximation per structure are 106\,s (ArchiveII) and 1039\,s (Eterna100).
The total runtimes are 106\,s and 1040\,s per structure, respectively.
Note that each motif bound can be cached once computed. As a result, the amortized runtime of \finaledit{LinearDecompose} can be significantly reduced when analyzing multiple structures sharing common motifs, which is often the case in practice.}

\vspace{-0.3cm}
\section{Discussion and Conclusion}
%


We introduced a theoretical framework based on ensemble approximation and probability decomposition to quantify the probabilistic designability of RNA secondary structures under \finaledit{any {\em factorizable} energy model  
or  scoring function that decomposes onto individual loops.}
We further develop a linear-time dynamic-programming algorithm that efficiently searches for optimal decompositions among exponentially many choices. \finaledit{Together, the resulting bounds offer a novel, interpretable characterization of probabilistic designability: rather than a black-box score, each bound is explicitly supported by rival motifs that thermodynamically dominate the target, and the optimal decomposition identifies the specific motifs responsible for low global designability.}

Applied to the ArchiveII and Eterna100 datasets, our algorithm {\finaledit{LinearDecompose}} consistently produces tighter probability bounds under the Turner energy model than existing methods. Moreover, the resulting optimal decompositions offer interpretable explanations of structural design difficulty by identifying the specific motifs that limit global designability.

 {\finaledit{LinearDecompose}} still has \finaledit{some} limitations:
\begin{enumerate}
  \item Unless a structure is designed with a probability close to the probability bound found by \finaledit{LinearDecompose}, the bound may be optimistic \finaledit{but} loose.
  \item Ensemble approximation is not universally applicable; for some structures or motifs, suitable rival motifs are hard to sample.
\end{enumerate}

Future work includes improving rival motif generation to tighten ensemble approximations {\finaledit{(see Appendix Sec.~\ref{sec:tightness} for more details)}} and leveraging the optimal decomposition given by {\finaledit{LinearDecompose}} as a guide for RNA design.

\vspace{-0.3cm}

\section*{\finaledit{Acknowledgement}}

\finaledit{This work was supported in part by NSF grant 2330737 (L.H.~and D.H.M.).} 
\finaledit{We thank the reviewers for comments.}

\bibliographystyle{natbib}
\bibliography{references}

\clearpage
\setcounter{section}{0}
\pagenumbering{roman}
\setcounter{figure}{0}
\renewcommand{\thefigure}{S\arabic{figure}}
\setcounter{table}{0}
\renewcommand{\thetable}{S\arabic{table}}
\renewcommand{\thesection}{S\arabic{section}}
\begin{center}
{\LARGE\bfseries Supplementary Information}
\end{center}

\section{Structural Loops}\label{sec:loops}

\begin{figure}[htbp]
	\centering
    \pdftooltip{\includegraphics[width=0.5\textwidth]{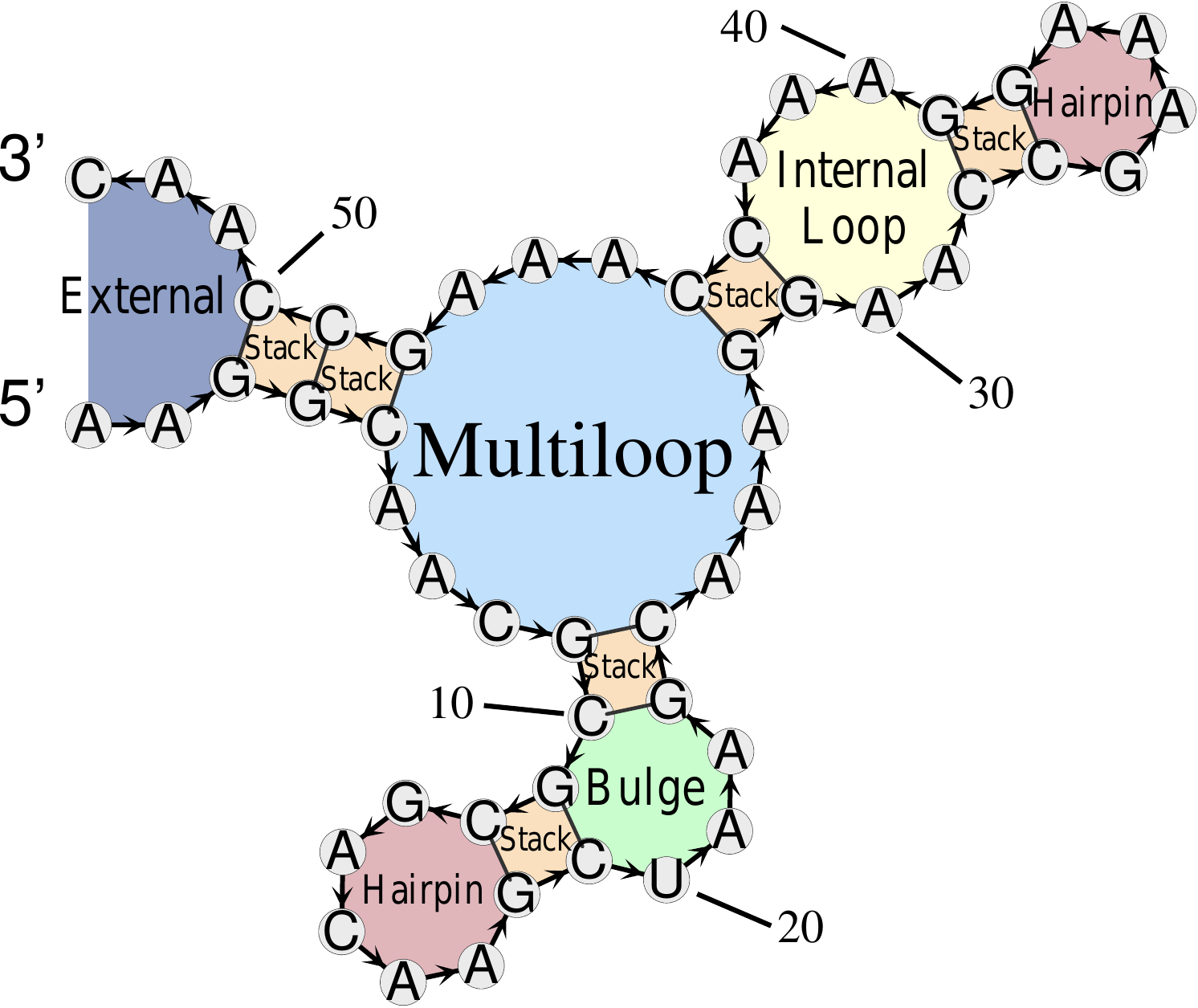}}{Diagram of an RNA secondary structure annotated with its loop types: hairpin (H), bulge (B), stack (S), internal loop (I), multiloop (M), and external loop (E), with arrows indicating 5' to 3' direction.}
    \captionof{figure}{An example of secondary structure and loops.} 
    \label{fig:loops_append}
\end{figure}

A secondary secondary structure can be decomposed into a collection of loops, where each loop is usually a region enclosed by base pair(s). Depending on the number of pairs on the boundary, main types of loops include hairpin loop, internal loop and multiloop, which are bounded by 1, 2 and 3 or more base pairs, respectively. In particular, the external loop is the most outside loop and is bounded by two ends ($5'$ and $3'$) and other base pair(s). Thus each loop can be identified by a set of pairs. Fig.~\ref{fig:loops_append}  showcases an example of secondary structure with various types of loops, where the some of the loops are notated as  

\begin{enumerate}
\item Hairpin: $H\blangle (12, 18) \brangle$.
\item Bulge: $B\blangle (10, 23), (11, 19) \brangle$.
\item Stack: $S\blangle (3, 50), (4, 49) \brangle$.
\item Internal Loop: $I\blangle (29, 43), (32, 39)\brangle$.
\item Multiloop: $M\blangle (5, 48), (9, 24), (28, 44)\brangle$.
\item External Loop: $E\blangle(3, 50) \brangle$.
\end{enumerate}

\begin{table}[ht]
        \centering
        \captionof{table}{Critical positions of loops in Fig.~\ref{fig:loop}}\label{tab:cr_ex}
        \begin{tabular}{|c|c|c|}
            \hline
            \multirow{2}{*}{Loop Type}  & \multicolumn{2}{|c|}{Critical Positions} \\ \cline{2-3}
            & Closing Pairs & Mismatches (Unpaired) \\ \hline
            External & 
            (3, 50) & 2, 51 \\
            Stack & (3, 50), (4, 49) & ~ \\
            Stack & (4, 49), (5, 48) & ~ \\
            Multi & (5, 48), (9, 24), (28, 44) & 4, 49, 8, 25, 27, 45 \\
            Stack & (9, 24), (10, 23) & ~ \\
            Bulge & (10, 23), (11, 19) & ~ \\
            Stack & (11, 19), (12, 18) & ~ \\
            Hairpin & (12, 18) & 13, 17 \\
            Stack & (28, 44), (29, 43) & ~ \\
            Internal & (29, 43), (32, 39) & 30, 42, 31, 40 \\
            Stack & (32, 39), (33, 38) & ~ \\
            Hairpin & (33, 38) & 34, 37 \\ \hline
        \end{tabular}
\end{table}

The function $\LP(\vecy)$ is used to denote the set of loops in a structure $\vecy$. The free energy of a secondary structure $\vecy$ is the sum of the free energy of each loop, 
\begin{equation}
  \DG(\vecx, \vecy) = \sum_{\vecz \in \loops(\vecy)} \DG(\vecx, \vecz),\notag
\end{equation}
The energy of each loop is typically determined by nucleotides on the positions of enclosing pairs and their adjacent mismatch positions, which are named as \emph{critical positions} in this article. Table~\ref{tab:cr_ex} lists the critical positions for all the loops in Fig.~\ref{fig:loop} and Table \ref{tab:cr} shows the indices of critical positions for each type of loops. Additionaly, some special hairpins of unstable triloops and stable tetraloops and hexaloops in Turner model have a separate energy lookup table. When evaluating the energy of a loop, it suffices to input only the nucleotides on its critical positions, i.e.,

\begin{equation}
 \DG(\vecx, \vecy) = \sum_{\vecz \in \loops(\vecy)} \DG(\vecx \proj \CR(\vecz), \vecz), \label{eq:e_sum}
\end{equation}

where  $\CR(\vecz)$ denotes the critical positions of loop $\vecz$ and  $\vecx\proj \CR(\vecz)$ denotes the nucleotides from $\vecx$ that are ``projected" onto $\CR(\vecz)$. 
The projection (\proj) allows us to focus on the relevant nucleotides for energy evaluation. For instance,
\begin{align}
  \CR(H\blangle (12, 18)\brangle) &= \{12, 13, 17, 18\},\\
  \CR(I\blangle (29, 43), (32, 39)\brangle) &= \{29, 30, 31, 32, 39, 40, 42, 43\}.
 \end{align}
For convenience, we also interchangeably put paired positions in brackets, i.e.,  
\begin{align}
 \CR(H\blangle (12, 18)\brangle) &= \{(12, 18), 13, 17\},\\
 \CR(I\blangle (29, 43), (32, 39)\brangle) &= \{(29, 43), (32, 39) , 30, 31, 40, 42\}.
\end{align}


\begin{table*}[t]
\centering
\captionsetup{justification=centering}

\caption{Critical positions for each type of loops under the Turner model implemented in ViennaRNA.
Special hairpins~\cite{Mathews+:2004} (triloops, tetraloops, and hexaloops) are not considered.}
\label{tab:cr}

\resizebox{.7\textwidth}{!}{%
\begin{tabular}{|c|c|c|}
\hline
\multirow{2}{*}{Loop Type} &
\multicolumn{2}{c|}{Critical Positions} \\ \cline{2-3}
& Closing Pairs & Mismatches \\ \hline
External
& $(i_1, j_1), (i_2, j_2), \ldots, (i_k, j_k)$
& $(i_1\!-\!1, j_1\!+\!1), \ldots, (i_k\!-\!1, j_k\!+\!1)$ \\
Hairpin
& $(i, j)$
& $i\!+\!1, j\!-\!1$ \\
Stack
& $(i, j), (k, l)$
& -- \\
Bulge
& $(i, j), (k, l)$
& -- \\
Internal
& $(i, j), (k, l)$
& $i\!+\!1, j\!-\!1, k\!-\!1, l\!+\!1$ \\
Multi
& $(i, j), (i_1, j_1), \ldots, (i_k, j_k)$
& $i\!+\!1, j\!-\!1, \ldots, i_k\!-\!1, j_k\!+\!1$ \\
\hline
\end{tabular}}
\end{table*}

\section{Projection Operation}\label{sec:proj}
\finaledit{The projection operation is a key step in our ensemble approximation algorithm, which allows us to focus on the relevant nucleotides for energy evaluation. The projection of a sequence $\vecx$ onto a set of positions $I$ is defined as follows:}
\begin{algorithm}[hbt!]
\editcolor
\caption{Projection $\hat{\vecx} = \vecx \proj I$}\label{alg:proj}
\begin{algorithmic}[1]
    \Function{Projection}{$\vecx, I$} \Comment{$I = [i_1, i_2, \ldots, i_n]$ is a list of critical positions}
    \State $\hat{\vecx} \gets \text{map}()$ \Comment{hash map}
    \For{$i$ in $I$}
        \State $\hat{\vecx}[i] \gets \vecx_i$ \Comment{Map  index $i$ to nucleotide $\vecx_i$}
    \EndFor
    \State \Return $\hat{\vecx}$
    \EndFunction
\end{algorithmic}
\end{algorithm}
\section{Structural \finaledit{Motif}}\label{sec:motif}


\begin{figure}[H]
    \centering

    \begin{minipage}[b]{0.25\textwidth}
        \centering
        \pdftooltip{\includegraphics[width=0.7\textwidth]{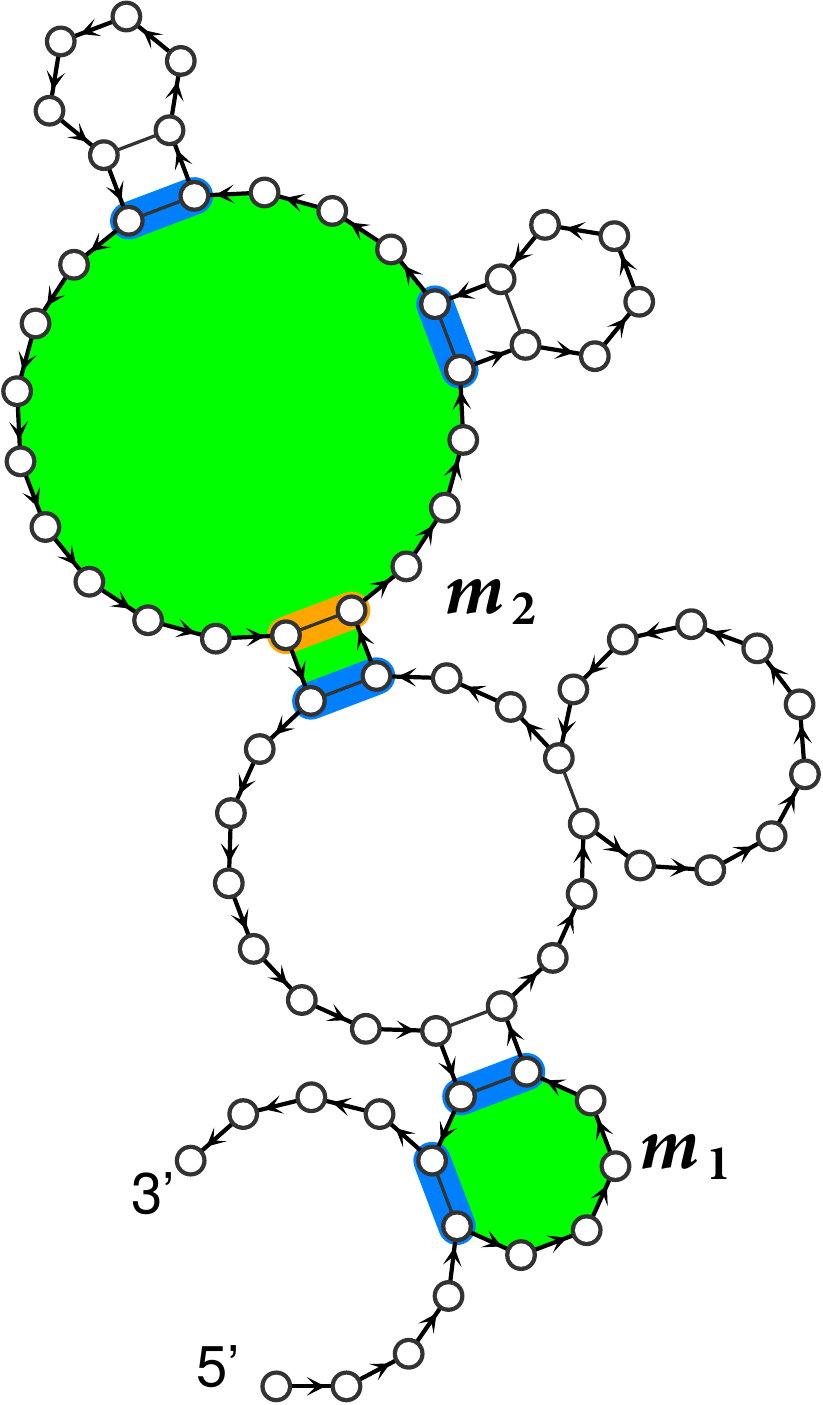}}{Motif m1 and m2 with cardinality 1 and 2 (one and two loops), showing loops in green, internal pairs in orange, and boundary pairs in blue.}
    \end{minipage}%
    \begin{minipage}[b]{0.25\textwidth}
        \centering
        \pdftooltip{\includegraphics[width=0.7\textwidth]{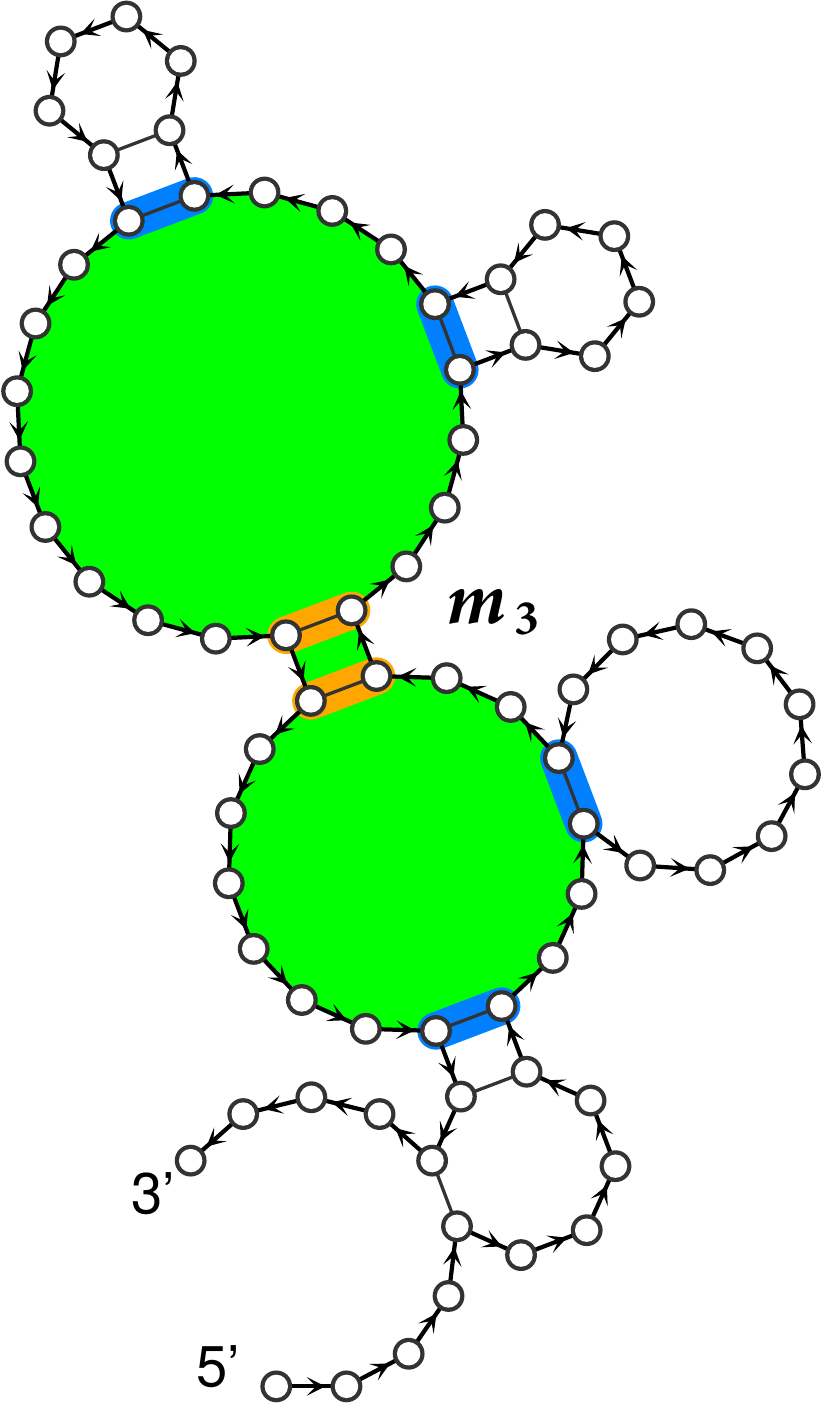}}{Motifs m3 with cardinality 3 (three loops), showing loops in green, internal pairs in orange, and boundary pairs in blue.}
    \end{minipage}

    \caption{Motifs of various cardinalities (numbers of loops):
    $\card(\M_1)\!=\!1$, $\card(\M_2)\!=\!2$, $\card(\M_3)\!=\!3$.
    Loops are highlighted in green, internal pairs ($\ipairs$) in orange and boundary pairs ($\bpairs$) in blue.}
    \label{fig:card}
\end{figure}

\subsection{Motif is a Generalization of Structure}\label{subsec:motif}
\begin{definition}\label{def:motif-loops}
A \m \vecm is a contiguous (sub)set of loops in an RNA secondary structure \vecy, notated $\vecm\subseteq\vecy$.
 \end{definition}
 
 Many functions defined for secondary structures can also be applied to motifs. For example, $\LP(\M)$ represents the set of loops within a motif \M, while $\pairs(\M)$ and $\unpaired(\M)$ represent the sets of base pairs and unpaired positions, respectively. We define the \emph{cardinality} of \vecm  as the number of loops in \vecm , i.e., $\card(\vecm) = |\loops(\vecm)|$. Fig.~\ref{fig:card} illustrates three motifs, $\M_1, \M_2$, and $\M_3$, in a structure adapted from the Eterna puzzle ``\texttt{Cat's Toy}''. These motifs contain 1, 2, and 3 loops, respectively. We also define the \emph{length} of a motif $|\vecm|$ as the number of bases it contains, which is consistent with the length of a secondary structure $|\vecy|$.

Since motifs are defined as sets of loops, we can conveniently use set relations to describe their interactions. A motif $\M_A$ is a \emph{sub-motif} of another motif $\M_B$ if $\M_A$ is contained within $\M_B$, denoted as $\M_A \subseteq \M_B$. For the motifs in Fig.~\ref{fig:card}, we observe the relation $\M_2 \subseteq \M_3$. We further use $\M_A \subset \M_B$ to indicate that $\M_A$ is a proper sub-motif of $\M_B$, meaning $\M_A \neq \M_B$. Therefore, $\M_2 \subset \M_3$. The entire structure $\vecy$ can be regarded as the largest motif within itself, and accordingly, $\M \subseteq \vecy$ signifies that motif \M is a part of structure $\vecy$, with $\M \subset \vecy$ implying \M is strictly smaller than \vecy. 

The loops in a motif \M are connected by base pairs. Each base pair in $\pairs(\M)$ is classified as either an \emph{internal pair} linking two loops in \M or a \emph{boundary pair} connecting one loop inside \M to one outside. These two types of pairs in \M are denoted as disjoint sets $\ipairs(\M)$ and $\bpairs(\M)$, respectively:
\begin{align}
\ipairs(\M) \cap \bpairs(\M) = \emptyset,~\ipairs(\M) \cup \bpairs(\M) = \pairs(\M).
\end{align}
Utilizing the commonly accepted nearest neighbor model for RNA folding, it becomes evident that certain motifs may be absent from structures folded from RNA sequences. For instance, motif $\M_3$ in Fig.~\ref{fig:card} is considered undesignable, as the removal of its two internal pairs consistently reduces the free energy. This brings us to the definition of an  \emph{\um}.

\subsection{Motif Ensemble from Constrained Folding}
The designability of motifs is based on \emph{constrained folding}. Given a sequence $\vecx$, a structure in its ensemble $\vecy \in \mathcal{Y}(\vecx)$, we can conduct constrained folding by constraining the boundary pairs of $\M$, i.e., $\bpairs(\vecm)$. We generalize the concept of (structure) \emph{ensemble} to \emph{motif ensemble} as the set of motifs  that \vecx can possibly fold into (under the constraint $\bpairs(\vecm)$ being forced), denoted as $\mathcal{M}(\vecx, bpairs(\vecm))$. 
In the context of constrained folding, the folding outcomes are unaffected by the nucleotides at the constrained positions. Thus, with slight notation abuse, we use $\vecx$ to denote a partial sequence corresponding to a motif $\vecm$, where each position in $\vecx$ matches a position in $\vecm$, and vice versa. By this definition, the motif ensemble of a partial sequence $\vecx$ is denoted as $\mathcal{M}(\vecx)$. Similarly, the notation $\mathcal{X}(\vecm)$ generalizes $\mathcal{X}(\vecy)$ and represents all (partial) RNA sequences whose motif ensembles contain $\vecm$.
Motifs in $\mathcal{M}(\vecx, $\bpairs(\vecm)$)$ have the same boundary pairs, i.e., 
\begin{equation}
\forall \vecm', \vecm'' \in \mathcal{M}(\vecx), \bpairs(\vecm') = \bpairs(\vecm'') = \bpairs(\vecm). 
\end{equation}

The \emph{free energy change} of a motif $\vecm$ is the sum of the free energy of the loops in \vecm,

\begin{equation}\label{eq:emotif}
\DG(\vecx, \vecm) = \sum_{\vecz\in \loops(\vecm)}  \DG(\vecx, \vecz).
\end{equation} 

The definitions of \MFE and \UMFE can also be generalized to motifs via constrained folding.
\begin{definition}
A motif $\mstar \subseteq \vecy$ is an \MFE motif of folding $\vecx$ under constraint $\bpairs(\vecm)$ , i.e., $\MFE(\vecx, \bpairs(\vecm))$, if and only if
\begin{equation}
\begin{aligned}
\forall \vecm \in \mathcal{M}(\vecx, \bpairs(\vecm))  \text{ and }\vecm \ne \mstar, \\ \DG(\vecx, \mstar) \leq  \DG(\vecx, \vecm). \label{def:mfe-motif}
\end{aligned}
\end{equation}

\end{definition}
\begin{definition} \label{def:umfe-motif}
A motif $\mstar \subseteq \vecy$ is an \UMFE motif of folding $\vecx$ under constraint $\bpairs(\vecm)$ , i.e., $\UMFE(\vecx, \bpairs(\vecm))$, if and only if 
\begin{equation}
\begin{aligned}
\forall \vecm \in \mathcal{M}(\vecx, \bpairs(\vecm))  \text{ and }\vecm \ne \mstar, \\ \DG(\vecx, \mstar) <  \DG(\vecx, \vecm). \label{eq:umfe-motif} 
\end{aligned}
\end{equation}
\end{definition}

Similarly, the equilibrium probability of a sequence folding into the motif is defined as, 
\begin{equation}
\begin{aligned}
p(\vecm \mid \vecx) &= \frac{e^{-\DG(\vecx, \vecm)/RT}}{Q(\vecx)} \\
&= \frac{e^{-\DG (\vecx, \vecm)/RT}}{\sum_{\vecm' \in \mathcal{M}(\vecx, \bpairs(\vecm))}e^{-\DG(\vecx, \vecm')/RT}}. \label{eq:prob_m}
\end{aligned}
\end{equation}

\editcolor


\section{Additional Pseudocode}\label{sec:alg}
\finaledit{The motif generation procedure used in \finaledit{LinearDecompose} is shown in Algorithm~\ref{alg:motif_gen}, which generates candidate motifs for each loop node based on the specified constraints of depth, width, and number of loops.}
\begin{algorithm}[t]
\caption{Constrained Motif Generation from a Loop}
\label{alg:motif_gen}
\begin{algorithmic}[1]
\Function{MotifGen}{$\eta, \maxdepths, \maxwidth, \maxloop$}
  \State $candidates \gets \emptyset$
  \For{$d \gets 1$ to \maxdepths}
    \State $\boldsymbol{M}_d \gets \emptyset$  \Comment{motifs of depth $d$ to be generated}
    \If{d = 1}
    	\State $\boldsymbol{M}_d \gets \{\eta\}$ \Comment{the motif is a single node $\eta$}
     \Else
          \For{$\vecm \in \boldsymbol{M}_{d-1} $}
                 \State $\vecm_{\text{new}} \gets \Call{Grow}{\vecm}$
                 \If{$\Call{Width}{\vecm_{\text{new}}} > \maxwidth$}
                 	\State \textbf{continue}
	         \EndIf
	         \If{$|\Call{Loops}{\vecm_{\text{new}}}| > \maxloop$}
                 	\State \textbf{continue}
	         \EndIf
	         \State $\boldsymbol{M}_d \gets \boldsymbol{M}_d \cup \vecm_{\text{new}}$
	  \EndFor
    \EndIf
    \State $candidates \gets candidates \cup \boldsymbol{M}_d$
  \EndFor
  \State \Return $candidates$
\EndFunction
\end{algorithmic}
\end{algorithm}

\section{Discussion on the Tightness of Ensemble Approximation}\label{sec:tightness}

The tightness of the probability bound produced by ensemble approximation (Algorithms~\ref{alg:approxe_1} and~\ref{alg:approxe_2}) depends on two key factors: the \emph{quality} and the \emph{number} of rival structures (motifs).

\subsection{Quality of rival structures.}
A rival structure is considered high quality if it thermodynamically dominates the target across a broad range of sequences, i.e., $\DDG(\vecx, \vecy', \ystar)$ is large and positive for many sequence assignments. Such rivals push the estimated bound down substantially. In contrast, a rival that only marginally outcompetes the target in a narrow region of sequence space contributes little to tightening the bound. In practice, rival quality is limited by the sampling strategy: rivals are obtained by folding sequences designed for the target, so they tend to be structurally similar to the target with small differential position sets $\D(\vecy', \ystar)$. When the sampled rivals are structurally distant or energetically weak competitors, the resulting bound may be loose.

\subsection{Number of rival structures.}
Including more rivals in the ensemble approximation directly tightens the bound, because each additional rival contributes an extra term to the denominator of the probability bound (Eq.~\ref{eq:rival}). Intuitively, more competitors in the approximated ensemble means the target structure claims a smaller share of the Boltzmann weight. The bound can only decrease as rivals are added, so the method is conservative by construction. However, adding more rivals also increases the size of the overall differential positions $\D(Y_r, \ystar)$ (Eq.~\ref{eq:dp_multiple}), which exponentially increases the enumeration cost. In practice, we impose an upper limit on the enumeration size, which may prevent the bound from fully benefiting from a large rival set.

\subsection{When bounds may be loose.}
Bounds tend to be loose when: (1) the sampled rivals are of low quality, e.g., they only weakly dominate the target; (2) the number of sampled rivals is small; or (3) the enumeration limit is reached before the best sequence assignment is found. These conditions are more likely for large structures with complex loop configurations, where the differential position set grows quickly and sampling diverse, high-quality rivals is harder.

\subsection{Utility for evaluating RNA design results.}
Beyond certifying undesignability, the probability bounds provide a practical tool for evaluating and comparing RNA design results. Table~\ref{tab:plots} illustrates varying degrees of bound tightness across different Eterna100 puzzles. For example, \emph{Simple Single Bond} achieves a best design probability of 0.370 against a bound of 0.394 — a tight gap indicating the structure is close to its theoretical limit and the design is near-optimal. In contrast, \emph{1, 2, 3 and 4 bulges} achieves only 0.005 against a bound of 0.303, suggesting substantial room for improvement. Similarly, \emph{multilooping fun} has a bound of 0.0006, consistent with the negligible best design probability of $2 \times 10^{-6}$, confirming that this structure is highly constrained and essentially undesignable at the ensemble level. These examples demonstrate that probability bounds provide a richer, quantitative characterization of design difficulty than binary undesignability criteria alone.

\section{Detailed Plots}
The following figures show the probability bounds vs.\ achieved probabilities for all 1144 ArchiveII structures and 100 Eterna100 structures.
In addition, the per-structure bounds from baselines (CountingDesign and CountingDesign+) are also scattered for comparison.
\begin{figure}[ht] 
\centering
  \pdftooltip{%
  \includegraphics[width=\linewidth]{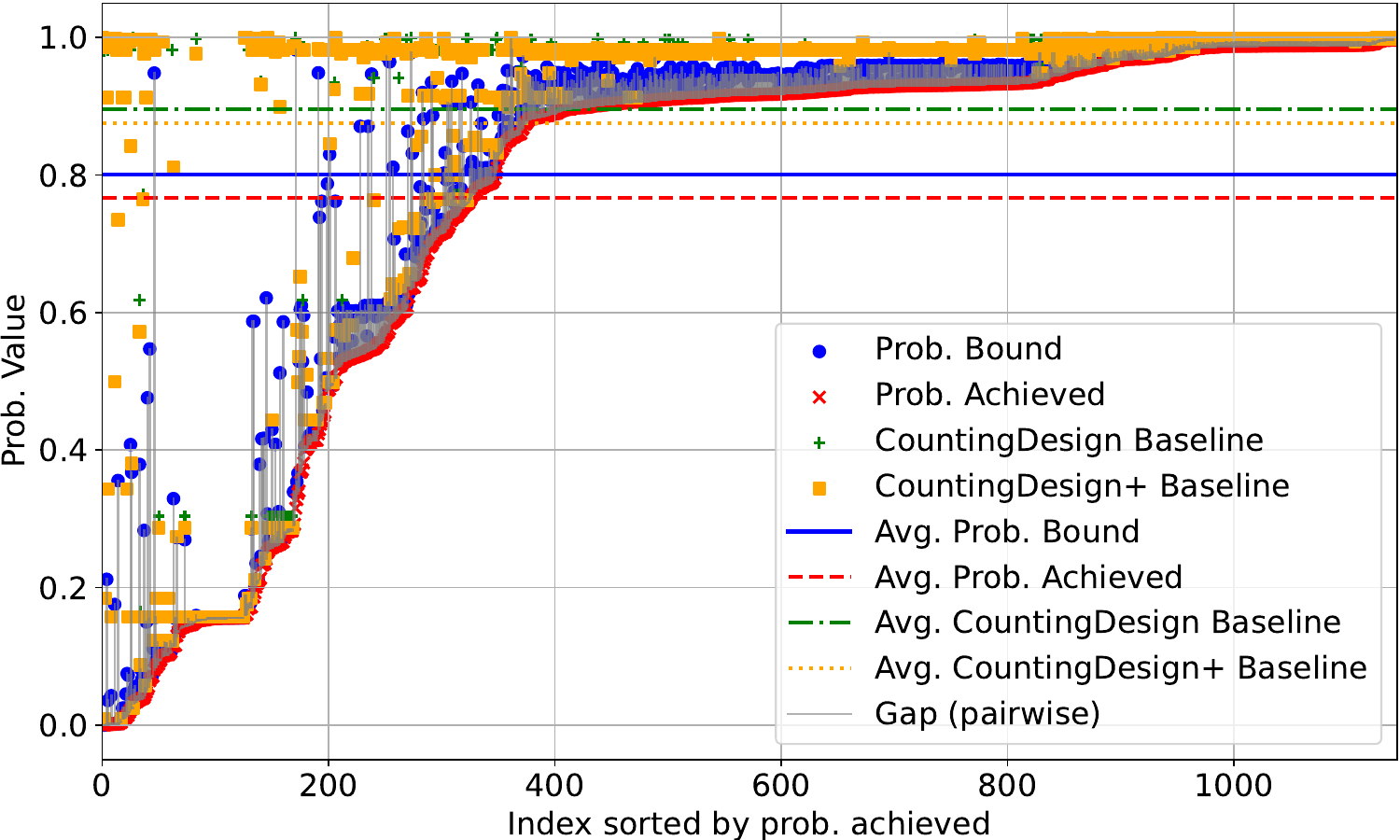}%
}{Scatter plot of probability bounds vs. achieved probabilities for 1144 ArchiveII structures. LinearDecompose yields tighter bounds than baselines.}
  \caption{Probability bounds vs.~achieved $p(\ystar\mid\vecx)$ on ArchiveII.}\
  \label{fig:bounds_archiveii_2}
%
  \centering
  \pdftooltip{%
  \includegraphics[width=\linewidth]{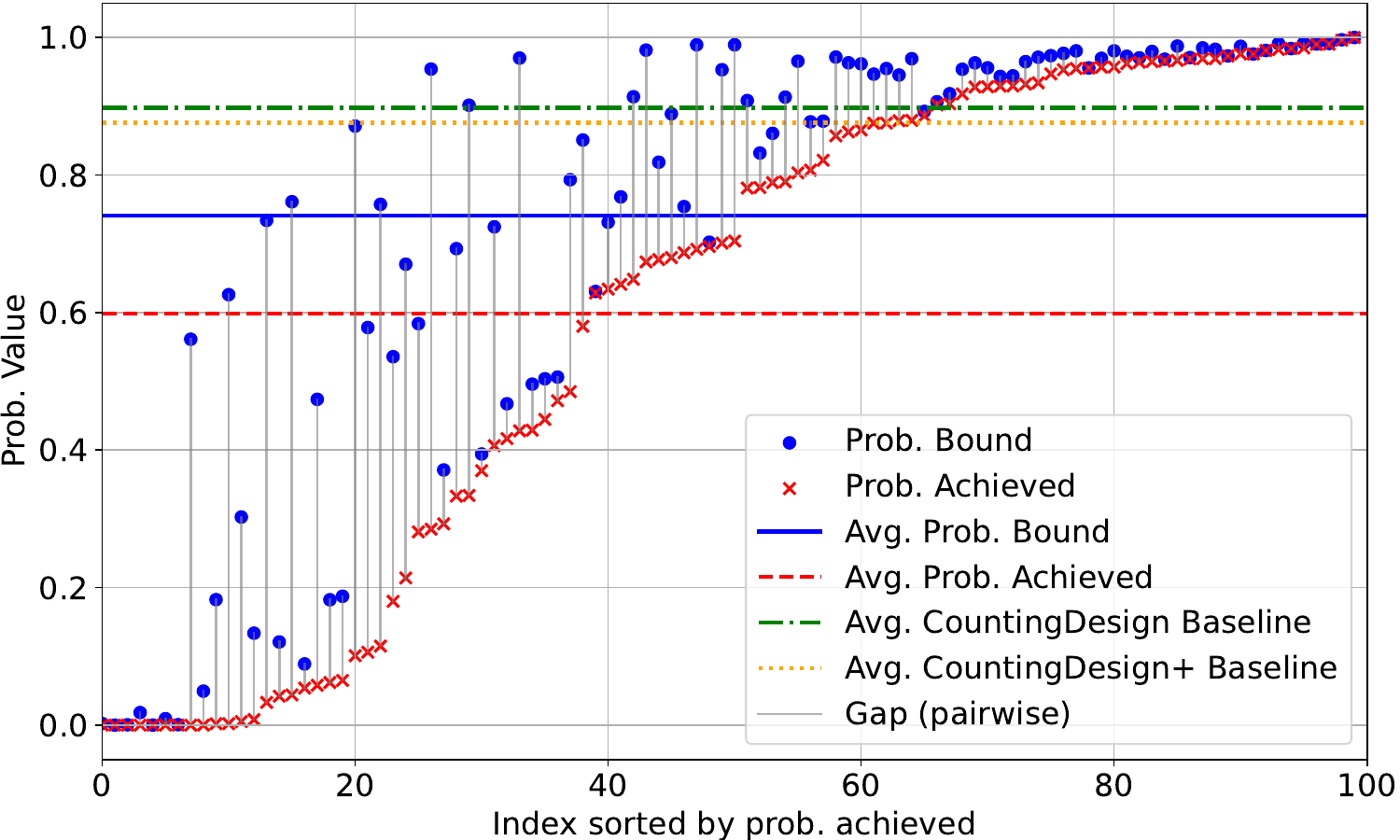}%
}{Scatter plot of probability bounds vs.\ achieved probabilities for 100 Eterna100 structures. LinearDecompose yields tighter bounds than baselines, with a larger gap between bounds and achieved probabilities compared to ArchiveII.}
  \caption{Probability bounds vs.~achieved $p(\ystar\mid\vecx)$ on Eterna100.}\
  \label{fig:bounds_eterna100_2}
  \vspace{-0.5cm}
\end{figure}

\end{document}